\documentclass[11pt]{article}

\newcommand{\focsarxiv}[2]{#2}

\pdfoutput=1
\usepackage[T1]{fontenc}
\usepackage{lmodern}
\usepackage{slantsc}
\usepackage[protrusion=true,expansion=true]{microtype}
\usepackage{breakcites}
\usepackage{amsmath,amssymb,amsfonts,amsthm}
\usepackage{subcaption}
\usepackage{graphicx}
\usepackage{fullpage}
\usepackage{setspace}
\usepackage[backref=page]{hyperref}
\usepackage[nameinlink]{cleveref}
\usepackage{color}
\usepackage{wrapfig}
\usepackage{tikz}
\usetikzlibrary{decorations.pathreplacing}
\usepackage{algorithm}
\usepackage[noend]{algpseudocode}
\usepackage[framemethod=tikz]{mdframed}
\usepackage{xspace}
\usepackage{pgfplots}
\usepackage{framed}
\usepackage{thmtools}
\usepackage{thm-restate}
\usepackage{tabu}
\usepackage{fancyhdr}
\usepackage{bbm}
\usepackage{todonotes}
\pgfplotsset{compat=1.5}

\newtheorem{theorem}{Theorem}[section]

\newtheorem{lemma}[theorem]{Lemma}
\newtheorem{proposition}[theorem]{Proposition}
\newtheorem{definition}[theorem]{Definition}

\newtheorem{claim}[theorem]{Claim}

\newenvironment{proofof}[1]{\begin{trivlist} \item {\bf Proof
#1:~~}}
  {\qed\end{trivlist}}

\newcommand{\namedref}[2]{\hyperref[#2]{#1~\ref*{#2}}}

\newcommand{\alglab}[1]{\label{alg:#1}}
\renewcommand{\algref}[1]{\namedref{Algorithm}{alg:#1}}

\renewcommand{\Re}{\operatorname{Re}}
\renewcommand{\Im}{\operatorname{Im}}

\def \FAIL    {\mdef{\mathsf{FAIL}}}

\def \CountSketch    {\mdef{\textsc{CountSketch}}}
\def \AMS    {\mdef{\textsc{AMS}}}

\newcommand{\PPr}[1]{\ensuremath{\mathbf{Pr}\left[#1\right]}}
\newcommand{\PPPr}[2]{\ensuremath{\underset{#1}{\mathbf{Pr}}\left[#2\right]}}
\newcommand{\Ex}[1]{\ensuremath{\mathbb{E}\left[#1\right]}}
\newcommand{\EEx}[2]{\ensuremath{\underset{#1}{\mathbb{E}}\left[#2\right]}}
\renewcommand{\O}[1]{\ensuremath{\mathcal{O}\left(#1\right)}}
\newcommand{\tO}[1]{\ensuremath{\tilde{\mathcal{O}}\left(#1\right)}}
\newcommand{\eps}{\varepsilon}

\def \calA    {\mdef{\mathcal{A}}}
\def \calB    {\mdef{\mathcal{B}}}

\def \calD    {\mdef{\mathcal{D}}}
\def \calE    {\mdef{\mathcal{E}}}

\def \calH    {\mdef{\mathcal{H}}}

\def \calN    {\mdef{\mathcal{N}}}

\def \calP    {\mdef{\mathcal{P}}}

\def \calS    {\mdef{\mathcal{S}}}

\def \calV    {\mdef{\mathcal{V}}}

\def \be    {\mdef{\mathbf{e}}}

\newcommand{\mdef}[1]{{\ensuremath{#1}}\xspace}  
\DeclareMathOperator*{\median}{median}

\DeclareMathOperator*{\argmax}{argmax}
\DeclareMathOperator*{\polylog}{polylog}
\DeclareMathOperator*{\poly}{poly}

\DeclareMathOperator*{\sign}{sign}

\DeclareMathOperator*{\Var}{Var}
\DeclareMathOperator*{\Poi}{Poi}

\DeclareMathOperator{\sech}{sech}

\newcommand{\abs}[1]{\mdef{\left|#1\right|}}         

\newcommand{\ignore}[1]{}

\newif\ifnotes\notestrue 
\ifnotes
\newcommand{\samson}[1]{\textcolor{blue}{{\bf (Samson:} {#1}{\bf ) }} \marginpar{\tiny\bf
             \begin{minipage}[t]{0.5in}
               \raggedright S:
            \end{minipage}}}
\newcommand{\david}[1]{\textcolor{purple}{{\bf (David:} {#1}{\bf ) }} \marginpar{\tiny\bf
             \begin{minipage}[t]{0.5in}
               \raggedright D:
            \end{minipage}}} 
\else
\newcommand{\samson}[1]{}
\newcommand{\david}[1]{}
\fi

\makeatletter
\renewcommand*{\@fnsymbol}[1]{\textcolor{mahogany}{\ensuremath{\ifcase#1\or *\or \dagger\or \ddagger\or
 \mathsection\or \triangledown\or \mathparagraph\or \|\or **\or \dagger\dagger
   \or \ddagger\ddagger \else\@ctrerr\fi}}}
\makeatother

\providecommand{\email}[1]{\href{mailto:#1}{\nolinkurl{#1}\xspace}}

\definecolor{mahogany}{rgb}{0.75, 0.25, 0.0}
\definecolor{darkblue}{rgb}{0.0, 0.0, 0.55}
\definecolor{darkpastelgreen}{rgb}{0.01, 0.75, 0.24}
\definecolor{darkgreen}{rgb}{0.0, 0.2, 0.13}
\definecolor{darkgoldenrod}{rgb}{0.72, 0.53, 0.04}
\definecolor{darkred}{rgb}{0.55, 0.0, 0.0}
\definecolor{forestgreenweb}{rgb}{0.13, 0.55, 0.13}
\definecolor{greencss}{rgb}{0.0, 0.5, 0.0}
\definecolor{bleudefrance}{rgb}{0.19, 0.55, 0.91}

\hypersetup{
     colorlinks   = true,
     citecolor    = mahogany,
	 linkcolor	  = darkpastelgreen
}

\fancypagestyle{pg}
{
\lhead{}
\rhead{}
\cfoot{--\ \thepage\ --}

}

\AtBeginDocument{%
  \DeclareFontShape{T1}{lmr}{m}{scit}{<->ssub*lmr/m/scsl}{}%
}

\begin{document}

\allowdisplaybreaks

\title{Perfect $L_p$ Sampling with Polylogarithmic Update Time}
\author{
William Swartworth\thanks{Carnegie Mellon University. 
E-mail: \email{wswartwo@andrew.cmu.edu}.
Supported in part by Office of Naval Research award number N000142112647 and a Simons Investigator Award.
}
\and
David P. Woodruff\thanks{Carnegie Mellon University and Google Research. 
E-mail: \email{dwoodruf@cs.cmu.edu}. 
Supported in part by Office of Naval Research award number N000142112647, a Simons Investigator Award, and NSF CCF-2335412.}
\and
Samson Zhou\thanks{Texas A\&M University. 
E-mail: \email{samsonzhou@gmail.com}. 
Supported in part by NSF CCF-2335411. 
The author gratefully acknowledges funding provided by the Oak Ridge Associated Universities (ORAU) Ralph E. Powe Junior Faculty Enhancement Award.}
}
\date{}

\maketitle
\thispagestyle{empty}

\begin{abstract}
Perfect $L_p$ sampling in a stream was introduced by Jayaram and Woodruff (FOCS 2018) as a streaming primitive which, given turnstile updates to a vector $x \in \{-\text{poly}(n), \ldots, \text{poly}(n)\}^n$, outputs an index $i^* \in \{1, 2, \ldots, n\}$ such that the probability of returning index $i$ is exactly \[\Pr[i^* = i] = \frac{|x_i|^p}{\|x\|_p^p} \pm \frac{1}{n^C},\] where $C > 0$ is an arbitrarily large constant. Jayaram and Woodruff achieved the optimal $\tilde{O}(\log^2 n)$ bits of memory for $0 < p < 2$, but their update time is at least $n^C$ per stream update. Thus an important open question is to achieve efficient update time while maintaining optimal space. For $0 < p < 2$, we give the first perfect $L_p$-sampler with the same optimal amount of memory but with only $\text{poly}(\log n)$ update time. Crucial to our result is an efficient simulation of a sum of reciprocals of powers of truncated exponential random variables by approximating its characteristic function, using the Gil-Pelaez inversion formula, and applying variants of the trapezoid formula to quickly approximate it.
\end{abstract}

\clearpage
\setcounter{page}{1}

\section{Introduction}
The \emph{streaming model of computation} has become a widely used framework for analyzing datasets too large to be stored in memory, such as transaction logs from commercial databases, Internet of Things (IoT) sensor measurements, social network activity, scientific observations, and virtual traffic monitoring. 
In the one-pass streaming model, a frequency vector over a universe mapped to the integers $[n]:=\{1,2,\ldots,n\}$ is implicitly constructed through a series of sequential updates to its components. 
Each update can be observed only once, and the objective is to compute, detect, or estimate some specific property of the input dataset while using memory that is sublinear in both the length of the data stream and the size of the input dataset. 

Sampling is one of the most fundamental and versatile techniques for analyzing large-scale datasets in the streaming model, and various forms of sampling~\cite{Vitter85,GemullaLH08,Muthukrishnan05,CohenDKLT11,CohenDKLT14} have been utilized in a number of different applications such as data summarization~\cite{FriezeKV04,DeshpandeRVW06, DeshpandeV07,AggarwalDK09,IndykMMM14,MahabadiIGR19, IndykMGR20,MahabadiRWZ20,cohen2020wor,MahabadiT23,GollapudiMS23}, database systems~\cite{LiptonNS90,HaasS92,Olken93,LiptonN95,HaasNSS96,GibbonsM98,Haas16,CohenG19}, distributed algorithms~\cite{TirthapuraW11,CormodeMYZ12,WoodruffZ12,CormodeF14,JayaramSTW19}, and anomaly detection in virtual networks~\cite{GilbertKMS01, EstanV03,MaiCSYZ06,HuangNGHJJT07,ThottanLJ10}. 
Formally, consider a vector $x \in \mathbb{R}^n$ that is updated through a sequence of $m$ operations. 
Each operation $t\in[m]$ modifies the coordinate $x_{i_t}$ by a value $\Delta_t \in\{-M,\ldots,-1,0,1,\ldots,M\}$, where $i_t\in[n]$ indicates the specific coordinate being updated and $M=\poly(n)$ is an upper bound on the magnitude of each update. 
In this way, the value of the coordinate $x_i$ at the end of the data stream can be aggregated as $x_i=\sum_{t\in [m]: i_t=i}\Delta_t$. 
We remark that since updates can both increase and decrease the coordinates of $x$, this is called the \emph{turnstile model}, whereas in the \emph{insertion-only model}, the updates are restricted to non-negative values, i.e., $\Delta_t \ge 0$ for all $t\in[m]$. 
The goal is to sample an index $i\in[n]$ with probability proportional to some fixed weight function $G(x_i)$ and possibly obtain an approximation of the value $x_i$. 

Perhaps the most popular choice for the function $G(\cdot)$ is the family of functions $G(z) = |z|^p$ for $p>0$, referred to as an \emph{$L_p$ sampler}, which have been used in tasks such as identifying heavy hitters, $L_p$ norm and $F_p$ moment estimation, cascaded norm approximation, and duplicate detection \cite{MonemizadehW10,AndoniKO11,JowhariST11, BravermanOZ12,JayaramW18,cohen2020wor,WoodruffZ21,JayaramWZ22,MahabadiWZ22,LinSWXZ25,WoodruffXZ25}. 
\begin{definition}[$L_p$-sampler]
Given a vector $x\in\mathbb{R}^n$, a weight parameter $p>0$, and a distortion parameter $\eps\ge 0$, an $\eps$-approximate $L_p$-sampler is an algorithm that selects an index $i^*\in[n]$ with probability at least $\frac{2}{3}$; otherwise it outputs a failure symbol $\bot$. 
If the algorithm outputs an index $i^*\neq\bot$, then for each $i\in[n]$,
\[\Pr[i = i^*] = \left(1\pm\eps\right) \cdot \frac{|x_i|^p}{\|x\|_p^p} \pm n^{-C},\]
where $C>0$ can be chosen to be any  constant. 
When $\eps = 0$, the sampler is said to be \emph{perfect}. 
\end{definition}
In the case of insertion-only streams, the well-known reservoir sampling technique~\cite{Vitter85} can produce an exact $L_1$ sample using $\O{\log n}$ bits of space, meaning $\eps=0$, and furthermore, the sampling probabilities are free from any additive $\frac{1}{\poly(n)}$ distortion. 
However, for $p \neq 1$ or in the case of turnstile streams, the problem becomes substantially more complex. 
Thus, the question of the existence of sublinear-space $L_p$ samplers was explicitly posed by Cormode, Murthukrishnan, and Rozenbaum~\cite{CormodeMR05}.

\cite{MonemizadehW10} first answered the question in the affirmative for \emph{approximate} $L_p$ samplers for $p\le 2$, giving a construction that uses $\poly\left(\frac{1}{\eps},\log n\right)$ bits of space. 
These dependencies were subsequently improved by Andoni, Krauthgamer, and Onak~\cite{AndoniKO11} and by Jowhari, Saglam, and Tardos~\cite{JowhariST11}, to roughly $\O{\frac{1}{\eps^{\max(1,p)}}\log^2 n}$ bits of space for $p\in(0,2)$ and $\O{\frac{1}{\eps^2}\log^3 n}$ bits of space for $p=2$. 
Although \cite{JowhariST11} showed a lower bound of $\Omega(\log^2 n)$ space for $p<2$, there was curiously no lower bound in terms of the distortion parameter $\eps$. 
This discrepancy was resolved by Jayaram and Woodruff \cite{JayaramW18}, who gave the first \emph{perfect} $L_p$ sampler that uses $\tO{\log^2 n}$ bits of space for $p\in(0,2)$ and $\O{\log^3 n}$ bits of space for $p=2$. 

Unfortunately, the algorithm of \cite{JayaramW18} requires polynomial time per update in the data stream. 
Specifically, a total variation distance of $n^{-c}$ requires roughly $n^{10c}$ time per update. 
For context, if the goal is to acquire roughly $64$ samples from even a relatively small dataset, then to simulate the desired distribution, the total variation distance would need to be roughly $2^{-6}$, which would necessitate $2^{60}\approx 10^{18}$ operations per update in the data stream. 
Though it should be noted that \cite{JayaramW18} did not attempt to optimize the update time, their techniques inherently require $\poly(n)$ operations on every single stream update. 
Hence, subsequent work has focused on achieving perfect $L_p$ samplers with fast optimal time. 
In particular, \cite{PettieW25} showed that there exists a perfect $L_p$ sampler that uses $\O{\log n}$ bits of space for $p\in(0,1)$ with $\O{1}$ time. 
However, their results require the random oracle model, and perhaps more crucially, only work for $p\le 1$ and the insertion-only model.  
By comparison, the classic reservoir sampling technique~\cite{Vitter85} already achieves a perfect $L_1$ sample using $\O{\log n}$ bits of space and $\O{1}$ update time for the insertion-only model. 
We ask:
\begin{quote}
Does there exist a perfect $L_p$ sampler for general $p \in (0,2)$ that uses small space and has fast update time in the turnstile streaming model (and without a random oracle)?
\end{quote}
Note that for $p > 2$, there is a lower bound of $n^{\Omega(1)}$ on the space complexity of any such algorithm, while for $p = 2$ the space complexity of known algorithms is $\O{\log^3 n}$, while the known lower bound is only $\Omega(\log^2 n)$. We focus on the regime $p \in (0,2)$, for which optimal $\tilde{\Theta}(\log^2 n)$ space bounds are known\footnote{Here the $\tilde{\mathcal{O}}$ notation suppresses $\poly(\log\log n)$ factors.}.  For a discussion on these space bounds, see \cite{kapralov2017optimal,JayaramW18}.

\begin{table}[!htb]
\centering
{
\tabulinesep=1.1mm
\resizebox{\columnwidth}{!}{
\begin{tabu}{|c|c|c|c|c|c|}\hline
Sampler & Data Stream & Distortion & Space & Update Time & Setting \\\hline
\cite{Vitter85} & Insertion-Only & Perfect & $\O{\log n}$ & $\O{1}$ & $p=1$\\\hline
\cite{MonemizadehW10} & Turnstile & Approximate & $\poly\left(\frac{1}{\eps},\log n\right)$ & $\polylog(n)$ & $p\in(0,2]$ \\\hline
\cite{JowhariST11} & Turnstile & Approximate & $\frac{1}{\eps^{\max(1,p)}}\cdot\polylog(n)$ & $\polylog(n)$ & $p\in(0,2]$ \\\hline
\cite{AndoniKO11} & Turnstile & Approximate & $\frac{1}{\eps^{\max(1,p)}}\cdot\polylog(n)$ & $\polylog(n)$ & $p\in(0,2]$ \\\hline
\cite{JayaramW18} & Turnstile & Perfect & $\O{\log^2(n)}$ & $\poly(n)$ & $p\in(0,2]$ \\\hline
\cite{PettieW25} & Insertion-Only & Perfect & $\O{\log(n)}$ & $\O{1}$ & $p\in(0,1]$ \\\hline\hline
Our Work & Turnstile & Perfect & $\O{\log^2(n)}$ & $\polylog(n)$ & $p\in(0,2)$ \\\hline
\end{tabu}
}
}
\caption{Summary of related work for $L_p$ sampling on data streams. Our main result provides an exponential improvement in the update time of \cite{JayaramW18}, which was the only previous perfect sampler in the turnstile model and which required poly$(n)$ time for every single update. Our space bounds suppress poly$(\log \log n)$ factors.}
\label{table:summary}
\end{table}

\paragraph{Our contribution.}
In this work, we answer the above question in the affirmative, giving a perfect $L_p$ sampler for $p\in(0,2)$ on turnstile streams that uses optimal space and has polylogarithmic update time, and does not require a random oracle. 
Formally, our main result is as follows:
\begin{theorem}
For any $p\in(0,2)$ and failure probability $\delta\in(0,1)$, there exists a perfect $L_p$ sampler on a turnstile stream that succeeds with probability at least $1-\delta$ and uses $\polylog(n)$ time per update. The algorithm uses $\tO{\log^2 n\log\frac{1}{\delta}}$ bits of space.
\end{theorem}
A thorough comparison of related work for sampling on data streams and our proposed samplers is displayed in \Cref{table:summary}.

\section{Technical Overview}\label{sec:tech}
A natural starting point for our algorithm is a common template followed by previous $L_p$ samplers~\cite{MonemizadehW10,AndoniKO11,JowhariST11,JayaramW18,MahabadiRWZ20,WoodruffXZ25} for turnstile streams. 
These algorithms perform a linear transformation on the input vector $x$ to obtain a new vector $z$. 
They then implement a heavy-hitter algorithm on $z$ to recover a vector $y$ and perform a statistical test on $y$. 
If the statistical test does not pass, then the algorithm does not report anything; otherwise if the statistical test passes, then the algorithm outputs the coordinate of $y$ with the largest magnitude as the sample. 
It can be shown that conditioned on the algorithm outputting an index $i^*\in[n]$, then the index $i^*$ follows the target distribution. 
Hence, the process is repeated a sufficiently large number of times to ensure that some instance succeeds with probability $1-\delta$. 

\paragraph{Linear transformation via uniform random variables.} 
A common choice~\cite{JowhariST11,MahabadiRWZ20} is to scale each coordinate $x_i$ by $\frac{1}{u_i^{1/p}}$, for independent uniform random variables $u_i\in[0,1]$. 
Observe that if $z_i=\frac{x_i}{u_i^{1/p}}$, then $|z_i|^p=\frac{|x_i|^p}{u_i}$ and so 
\[\PPr{|z_i|^p\ge\|x\|_p^p}=\PPr{u_i\le\frac{|x_i|^p}{\|x\|_p^p}}=\frac{|x_i|^p}{\|x\|_p^p},\]
which exactly follows the desired $L_p$ sampling distribution. 

However, there are many challenges with implementing this idealized process in the streaming model. 
One major issue is that the algorithm cannot exactly compute $|z_i|^p$ and $\|x\|_p^p$.  
Moreover, any approximation to these quantities immediately induces undesirable distortion in the sampling probabilities. 
For example, a heavy-hitter algorithm that guarantees a multiplicative $2$-approximation to $|z_i|^p$ may produce an estimate $\widehat{|z_i|^p}=2|z_i|^p$ and another estimate $\widehat{|z_j|^p}=|z_j|^p$. 
Then using $\widehat{|z_i|^p}$ and $\widehat{|z_j|^p}$ instead of the true values $|z_i|^p$ and $|z_j|^p$ can result in a distortion of up to $2$ in the sampling probabilities, so that the resulting procedure would only be an $2$-approximate $L_p$ sampler in this case.
In fact, the sampling probabilities can be distorted by larger factors, since heavy-hitter algorithms such as $\CountSketch$~\cite{CharikarCF04} have additive error as large as $\|z\|_2$, which can be significantly larger than $|z_i|^p$. 
Another issue is that there can be multiple indices $i\in[n]$ such that $|z_i|^p\ge\|x\|_p^p$. 

\paragraph{Linear transformation via exponential random variables.} 
Instead, suppose in the above argument based on uniform random variables that there exists a single index $i^*$ for which  $|z_{i^*}|^p\ge\|x\|_p^p$. 
Conditioned on this event, then by the above, it suffices to identify the index $i^*$, which is the largest coordinate in absolute value in $z$, as this will give the desired sample. 
It can be shown that $z_{i^*}$ is a heavy-hitter of $z$, so that $z_{i^*}$ can be estimated within a constant-factor approximation  in a turnstile stream. 
This intuition motivates using a different linear transformation~\cite{AndoniKO11,JayaramW18}.  
Instead of using uniform random variables, let $z_i=\frac{x_i}{e_i^{1/p}}$, for independent exponential random variables $e_i$. 
Let $i^*$ be the index for which $|z_{i^*}|$ is the maximum. 
By the properties of exponential random variables, it is well-known that for any $i \in \{1, 2, \ldots, n\}$ 
$$\Pr[i^* =i] = \frac{|x_i|^p}{\|x\|_p^p},$$
and thus, if we could identify $i^*$, then this would be a perfect $L_p$ sample. 
The use of exponentials for sampling from large data sets is now a standard idea behind a number of sampling schemes, see, e.g., \cite{Cohen1997SizeEstimation,CohenDKLT11,CohenDKLT14,CohenG19,Cohen2017HyperLogLog,Cohen2018StreamSampling,JayaramW18,Cohen2023SamplingBigIdeas}. 

Another standard and important point is that $\max_{i\in[n]}\frac{x_i}{e_i^{1/p}}\sim\frac{\|x\|_p}{E^{1/p}}$, where $E$ is another exponential random variable, and thus it can be shown that $z_{i^*}=\max_{i\in[n]} z_i=\Omega(\|x\|_p)$ with constant probability. Therefore, using a heavy-hitters algorithm with $\Omega(1)$ threshold, we can well-approximate $z_{i^*}$. There is still estimation error for $z_{i^*}$, but if the error is smaller than the gap between $z_{i^*}$ and $\max_{i\in[n]\setminus i^*} z_i$, then the algorithm still succeeds in identifying the maximum coordinate. It is possible to design a statistical test to detect this case, which suffices for an {\it approximate} $L_p$ sampler, though not a {\it perfect} $L_p$ sampler. 

Unfortunately, estimation error from heavy-hitter algorithms can be a function of the coordinates themselves. 
For example, if $x=(n^{1/p},1,1,1,\ldots)$, then depending on the heavy-hitter algorithm, it can be shown that the estimation error is significantly larger when $z_2$ is the maximum coordinate, as opposed to when $z_1$ is the maximum coordinate. 
Informally, this is because when $z_2$ is the maximum coordinate, then $x_1$ can contribute to the error term, resulting in a much larger error than if $z_1$ is the maximum coordinate. 
More generally, the statistical test can depend on the identity of the anti-rank vector $D(1),\ldots,D(n)$, where $D(k)$ is the index that corresponds to the $k$-largest value of $z_i$, i.e., $z_{D(1)}\ge z_{D(2)}\ge\ldots$, with ties broken arbitrarily. 
Thus, this approach is insufficient for perfect $L_p$ samplers, because the failure probability due to the statistical test can itself change by an $\Omega(1)$ factor, depending on the identity of the maximum coordinate $D(1)$. 

\paragraph{Linear transformation via duplication.}
One possible approach to remove the dependencies on the anti-rank vector is to duplicate the coordinates $\poly(n)$ times, an approach introduced by \cite{JayaramW18}. 
Instead of only scaling each coordinate $x_i$ by an exponential random variable $\frac{1}{e_i^{1/p}}$, we can first duplicate all coordinates $x_i$ to remove heavy items from the stream. 
Specifically, we define the duplicated vector $V$ so that the first $n^c$ entries of $V$ are all equal to $x_1$, the next $n^c$ entries of $V$ are all equal to $x_2$, and so forth. 
Hence, we have $\frac{|x_i|}{\|V\|_1}\le\frac{1}{n^c}$ for all $i\in[n]$, which informally reduces the dependence on the anti-rank vector since conditioned on the maximum, there are still $n^C-1$ copies of it that could be the second maximum.  
Formally, the order statistics $z_{D(k)}$ are highly concentrated around their expectation, which can be expressed as a sum of hidden independent exponential random variables $E_1,E_2,\ldots$~\cite{nagaraja2006order} that are independent of the anti-rank vector. 

Although the resulting universe has size $N=n\cdot n^c$, heavy-hitter algorithms will still only use $\polylog(N)=\polylog(n)$ space. 
On the other hand, all $n^c$ copies of a single coordinate $i\in[n]$ must be updated in the duplicated vector $V$ because they are scaled by independent exponential random variables $e_{i\cdot n^c+1}^{-1/p},e_{i\cdot n^c+2}^{-1/p},\ldots,e_{(i+1)\cdot n^c}^{-1/p}$. 
As a result, the update time for perfect sampling is  $\Omega(n^c)$, which is prohibitively slow. 

\paragraph{Heavy-hitters through dense Gaussian sketches.}

The runtime bottleneck of the current framework is due to the heavy-hitter algorithm on the duplicated vector $V$ for the purposes of identifying $i^*=\argmax_{i\in[N]} V_i$. 
In particular, existing approaches based on sparse variants of $\CountSketch$~\cite{CharikarCF04} hash each coordinate $V_i$ to a bucket $h(i)$ and maintain a signed sum $\sum_{j:h(j)=h(i)}s_jV_j$ of the elements hashed to the bucket, where $s_j\in\{-1,+1\}$ are random signs for all $j\in[N]$.  
This can be viewed as a linear sketch where each row of the sketch matrix corresponds to a separate bucket. 
Due to the dependencies of the same coordinate $V_i$ being hashed to exactly one bucket, i.e., one row, it seems challenging to simulate this distribution without explicitly generating each of the exponential random variables $e_1,\ldots,e_N$ and explicitly performing the hashing oneself, and thus spending polynomial update time. 

Instead, we consider a dense version of $\CountSketch$, which estimates the entries of $V\in\mathbb{R}^N$ by a random linear sketch $G\in\mathbb{R}^{k\times N}$ whose entries are i.i.d. $N(0,1/k)$ Gaussian random variables, for a certain value of $k$. 
The algorithm maintains $G \cdot V$ throughout the data stream. If $G_i$ denotes the $i$-th column of $G$, then $\langle G_i, GV\rangle$ is the estimate for $V_i$. 

The main benefit of the dense Gaussian $\CountSketch$ is its $2$-stability, that is, if we can efficiently simulate $\|V^i\|_2$, where $V^i$ is the vector of the $n^C$ reciprocals of $1/p$-th powers of exponentials associated with original coordinate $i$, then given an update to the $i$-th coordinate $x_i \leftarrow x_i + \Delta$, we could add $\Delta \cdot N_{i, \ell}$ to the $\ell$-th row of $GV$, where $N_{i, \ell}$ is distributed as $N(0, \|V^i\|_2^2)$. Note that we generate the $N_{i,\ell}$ once for all $i$ and $\ell$ and use them for all updates - see below on how this is derandomized memory efficiently in a stream. Had we instead used the standard $\CountSketch$, we would have that some duplicate coordinates associated with the same coordinate $x_i$ may be hashed to two separate buckets, which then causes a non-trivial dependence if one has $\O{\log n}$ rows of CountSketch since, having generated the norm of exponential scalings for each bucket in the first row, one must now simulate a distribution which re-partitions the same exponential scalings across different buckets in the second row, which seems highly non-trivial. Instead by using the dense Gaussian $\CountSketch$, after $\|V^i\|_2$ has been generated, one can simply sample independent Gaussians to generate the $N_{i,\ell}$ random variables. 

%
%

\paragraph{Central limit theorem for sums of heavy-tailed random variables.}
Now observe that $\|V^i\|_2^2=\sum_{j\in[n^c]}\frac{x_i^2}{e_{i,j}^{2/p}}$. 
It can be shown that for an exponential random variable $E$, the tail of the random variable $X=\frac{1}{E^{2/p}}$ follows the probability density function $\Theta\left(\frac{1}{x^{p/2+1}}\right)$, which behaves roughly like the tail of a $\frac{p}{2}$-stable random variable~\cite{Zolotarev89,Nolan03}. 
Since the sum of stable random variables is another stable random variable with different parameters, this suggests analyzing the rate of convergence of $\sum_{j\in[n^c]}\frac{x_i^2}{e_{i,j}^{2/p}}$ to a $\frac{p}{2}$-stable random variable as $C$ increases. 
Note that it is not sufficient to show the convergence in distribution in the limit, because $n^C$ is the universe size and the resulting space complexity is $\polylog(n^C)$, whose asymptotic behavior changes if $C=\omega(1)$. 
Fortunately, it can be shown~\cite{vogel2022quantitative} that the probability density function at sufficiently large values $x$ between the two distributions are point-wise within a multiplicative $\left(1+\frac{1}{n^{\O{C}}}\right)$, so that the total variation distance is $\frac{1}{\poly(n)}$ for sufficiently large $C$. 

More seriously though, this does \emph{not} translate to an algorithm because this does not allow us to efficiently determine the {\it max entry} of $V^i$, or an approximation to it, from the linear sketch $GV$. 
That is, conditioned on the value of $\|V^i\|_2^2$ that we have quickly simulated, it is not clear how to generate consistent values of the $\frac{x_i}{e_{i,j}^{2/p}}$, which are needed to provide estimates of $V^i$ from the heavy-hitters algorithm, and thus we cannot perform a statistical test that compares the maximum scaled coordinate value to $\|V^i\|_2$. 

%

\paragraph{Performing the sampling.}
From the above discussion, we can see that what is important is the ability to generate the {\it joint} distribution of the maximum of a list of reciprocals of $1/p$-th powers of exponentials, together with the squared norm of the remaining entries in the list. 

Suppose that $E$ is an exponential, and that $X_1 > X_2 > \ldots > X_k$ are sampled by choosing $k$ i.i.d. samples according to the distribution of $E^{-2/p}$, and then arranged in decreasing order.  We would like to sample from the joint distribution
$(X_1, X_2 + \ldots + X_k),$ so that we would have both the max as well as the overall norm needed in order to perform our statistical test. By max stability, $X_1$ is distributed as $k^{2/p}(\text{exponential}(1))^{-2/p}$, so it is convenient to normalize to
\[
k^{-2/p}(X_1, X_2 + \ldots + X_k).
\]
The first coordinate is now always $\text{exponential}(1)^{-2/p}$, and this can easily be sampled. 
The terms in the sum would have $p/2$-stable tails, except we must truncate each of the $X_i$'s at $X_1.$ If the $X_i$'s did not need to be truncated, then the second coordinate would converge in distribution to a $(p/2)$-stable distribution by the Generalized Central Limit Theorem.  
With truncation, it converges to a similar, but slightly perturbed distribution.  
Unfortunately this distribution is quite complicated and does not seem to be expressible in terms of distributions of well-known random variables.  
However, we can write down its characteristic function precisely by computing the characteristic function of each $X_j$, $j \geq 2$, and taking their product. 
Given this, the Gil-Pelaez Fourier inversion formula~\cite{wendel1961non} allows us to express the value of the cumulative density function (CDF) $F$ at $t $ for this distribution as a Fourier-type integral.  
Since the resulting integral is that of an analytic function, it turns out that techniques from numerical analysis and in particular the trapezoid rule allow us to evaluate the integral to precision $\eps$ in only $\polylog(1/\eps)$ time. 
The calculations required here turn out to be quite involved and non-standard; we refer the reader to Theorem \ref{thm:calculation} in Section \ref{sec:calculation} for further details. 

Finally, given the ability to evaluate $F(t)$, we may simply choose a uniformly random $y \in [0,1]$ and binary search for the value $x$ with $F(x) = y.$  
This gives a value $x$ that (up to rounding error) is very nearly sampled from the desired distribution for $X_2+\ldots+X_k$. 

\paragraph{Achieving worst-case time.} 
Unfortunately, in some cases we may need to evaluate $F(t)$ on a very large argument $t$. 
The time of our evaluation procedure depends polynomially on $t$, $R$, and $1/R$, where $R$ is the ratio of $X_1$ to $\sum_{j > 1} X_j$. 
In our algorithm we can preprocess our list by first generating not only the maximum $X_1$, but actually the top $\tau=\O{\log n}$ order statistics using standard formula for order statistics of exponential random variables. 
By the memorylessness property of exponentials, after generating $X_1, \ldots, X_\tau$, we are tasked with the same problem of generating the sum of remaining $X_j$ conditioned on each exponential being at most $X_\tau$, so we can use the same procedure as before. 
This ensures $R = \O{1}$ with failure probability $1/\poly(n)$. 
Similarly, by standard properties of exponential random variables, the ratio of $X_\tau$ to $\sum_{j > 1} X_j$ is $\Omega(1/\log n)$ with failure probability $1/\poly(n)$, so this ensures $1/R = \O{\log n}$. 
Finally, we are evaluating our CDF at the normalized average (since we look at the limiting distribution) of the $X_j$, $j > \tau$, and by level set arguments we can show that the argument $t$ we evaluate at is at most $\poly(1/R, R, \log n)$ with failure probability $1/\poly(n)$. Combining these, we achieve $\poly(\log n)$ worst case update time. 

\paragraph{Statistical test.}
Due to the behavior of the limiting distribution, we cannot immediately apply statistical tests that consider the order statistics of the distribution. 
Instead, we view the vector $z$ as having dimension $n'=\O{n\log n}$ rather than $N=n^{C+1}$, so that the coordinates of $z$ are the values $\{\calH_i\}$ and $N_i$ for all $i\in[n]$. 
Here, $\{\calH_i\}$ is the set of the top $\O{\log n}$ reciprocals of powers of exponential random variables and $N_i$ is the linear combination of the remaining variables. 
Thus, due to the consolidation of the linear combination of the duplications that do not attain the maximum, the vector $z\in\mathbb{R}^{n'}$ serves as a compact representation of the previous vector $z\in\mathbb{R}^N$. 
We use the dense $\CountSketch$ data structure to compute estimates $\widehat{z_{D(1)}}$ and $\widehat{z_{D(2)}}$ of the two largest coordinates $z_{D(1)}$ and $z_{D(2)}$. 
We compute a $2$-approximation $Z$ of $\|z\|_2$ and generate a fixed constant $\eps>0$, along with a uniform random variable $\mu$ drawn from the interval $[0.99, 1.01]$. 
The statistical test then fails if $\widehat{z_{D(1)}}-\widehat{z_{D(2)}}<100\mu\eps Z$ or $\widehat{z_{D(2)}}<50\mu\eps Z$; otherwise, the algorithm outputs $D(1)$. 
We remark that the purpose of $\eps$ is to ensure that the test passes with sufficiently high constant probability while the purpose of the uniform random variable $\mu$ is to remove the dependencies on the maximum index and the infinite duplication. 

\paragraph{Derandomization.}
We have so far assumed that our random variables can be freely generated and accessed. 
To convert this into a true streaming algorithm, we need to derandomize the components of the linear sketch. 
The standard approach is to use Nisan's pseudorandom generator (PRG)~\cite{Nisan92}, which informally states that any algorithm which uses $S$ space and $R$ random bits can be derandomized up to $2^{-S}$ total variation distance by an algorithm which uses only $\O{S\log R}$ bits of space. 
However, we require $R=\poly(n)$ random variables already for the Gaussians in the dense $\CountSketch$, and so applying Nisan's PRG would result in an undesirable $\O{\log n}$ multiplicative overhead in space. 

Instead, as in \cite{JayaramW18}, we use the PRG construction of \cite{GopalanKM18}, which is designed to fool a broad class of low-complexity functions known as \emph{Fourier shapes}. 
Roughly speaking, a Fourier shape is a product of bounded one-dimensional complex-valued functions, and the PRG of \cite{GopalanKM18} shows that all such functions can be $\eps$-fooled using a seed of length $\tO{\log\frac{n}{\eps}}$. 
Because the randomized components of our sketch, both the Gaussian projections and the exponential scaling, depend on independent random variables through such product structures, this generator is a natural fit for our setting. 


To make this formal, we note that our algorithm relies on two sources of randomness. 
The first, denoted $R_g$, corresponds to the Gaussian variables used in the dense $\CountSketch$. 
The second, denoted $R_e$, corresponds to the exponential random variables used for scaling in the $L_p$ sampler, including those used in the simulation of the top $\tau$ order statistics and the remaining tail statistics. 
It can be shown that we need $\poly(n)$ bits of randomness for both sources. 
For any fixed $R_e$, we define a tester $\calA_{i,R_e}$ for each coordinate $i\in[n]$ that determines whether the sampler selects index $i$ given the Gaussian randomness $R_g$. 
Using the PRG of \cite{JayaramW18}, we can generate a deterministic sequence that effectively fools $\calA_{i,R_e}$ for all $i$, replacing $R_g$ with a pseudorandom stream produced from a short seed.
Similarly, the randomness $R_e$ used for the exponential variables can be derandomized by simulating the largest order statistics and the remaining tail statistics deterministically. 
The tail statistics are generated using a fixed procedure that draws from the cumulative distribution function in a way that can be expressed as a small number of Fourier-shape queries. 
By applying the PRG of \cite{GopalanKM18} to this setting, we obtain a deterministic generator for $R_e$ that preserves the correctness of the sampler. 

For any fixed $i \in [n]$, the algorithm can be viewed as a tester $\calA_i(R_g, R_e) \in \{0, 1\}$ that determines whether the $L_p$ sampler selects index $i$ as the sampled coordinate. 
We then use a hybrid argument to show that there exist two instantiations $F_1(\cdot)$ and $F_2(\cdot)$ of the PRG so that
\focsarxiv{
\begin{align*}
&\left\lvert\PPPr{R_g,R_e}{\calA_i(R_g,R_e)=1}-\PPPr{x,y}{\calA_i(F_1(x),F_2(y))=1}\right\rvert\\
&\hspace{1in}\le\frac{1}{n^C},
\end{align*}
}
{
\[\left\lvert\PPPr{R_g,R_e}{\calA_i(R_g,R_e)=1}-\PPPr{x,y}{\calA_i(F_1(x),F_2(y))=1}\right\rvert\le\frac{1}{n^C},\]
}
for any arbitrarily large constant $C$, where $y_1$ and $y_2$ are seeds of length $\tO{\log^2 n}$. 
  
\section{Preliminaries}
For any integer $n > 0$, we denote the set $\{1, \ldots, n\}$ by $[n]$. 
The notation $\poly(n)$ refers to a polynomial function of $n$, while the notation $\polylog(n)$ represents a polynomial function of $\log n$. 
An event is said to occur with high probability if its probability is at least $1-\frac{1}{\poly(n)}$.

We state the following Gaussian Khintchine-type inequality, generalizing from Rademacher random variables. 
\begin{theorem}[Marcinkiewicz–Zygmund inequality]
\label{thm:khintchine}
\cite{marcinkiewicz1938quelques}
Let $r_1,\ldots,r_n$ be independent normal random variables and let $p\ge 1$. 
Then there exist absolute constants $A_p,B_p>0$ that only depend on $p$, such that
\[(A_p)^p\cdot\|x\|_2^p\le\Ex{|r_1x_1+\ldots+r_nx_n|^p}\le(B_p)^p\cdot\|x\|_2^p.\]
\end{theorem}
We remark that the Marcinkiewicz–Zygmund inequality more generally applies to random variables $r_1,\ldots,r_n$ with mean zero and finite variance, but for our purposes, normal random variables suffice. 

Next, we recall a dense variant of the standard $\CountSketch$ algorithm for finding $L_2$-heavy hitters, c.f., \Cref{alg:dense:countsketch}. 
We first generate a random matrix $G$ with $k$ rows and $n$ columns, where each entry is drawn from a normal distribution $\calN(0,1)$. 
Here, $k$ is an error parameter, so that a larger number of rows correspond to a smaller error. 
The estimate for each $x_i$ is obtained by taking the dot product of the corresponding column $G_i$ with the transformed frequency vector $Gx$. 
Finally, the result is scaled by a factor of $\frac{1}{k}$ to account for the $k$ separate rows. 

\begin{algorithm}[!htb]
\caption{Dense $\CountSketch$}
\label{alg:dense:countsketch}
\begin{algorithmic}[1]
\Require{Frequency vector $x\in\mathbb{R}^n$}
\Ensure{Estimates of $x_i$ for all $i\in[n]$}
\State{Let $G\in\mathbb{R}^{k\times n}$ be a matrix where each entry is drawn from $\calN(0,1)$}
\State{Let $G_i$ denote the $i$-th column of $G$}
\State{\Return $\widehat{x_i}=\frac{1}{k}\cdot\langle G_i,Gx\rangle$ as the estimate for $x_i$}
\end{algorithmic}
\end{algorithm}
We having the following guarantees for dense $\CountSketch$ in \Cref{alg:dense:countsketch}. 
We remark that the proof is completely standard. 
\begin{lemma}
\label{lem:dense:countsketch}
For each $i\in[n]$, we have 
\[\PPr{\left\lvert\widehat{x_i}-x_i\right\rvert\le\frac{4}{\sqrt{k}}\cdot\|x\|_2}\ge\frac{3}{4}.\]
\end{lemma}
\begin{proof}
Consider a fixed $i\in[n]$. 
Then
\focsarxiv{
\begin{align*}
\widehat{x_i}
&=\frac{1}{k}\cdot\langle G_i,Gx\rangle\\
&=\frac{1}{k}\sum_{a=1}^k G_{a,i}\left(\sum_{j=1}^n G_{a,j}x_j\right)\\
&=\frac{1}{k}\sum_{a=1}^k\left(G_{a,i}^2x_i+\sum_{j\ne i}G_{a,i}G_{a,j}x_j\right).
\end{align*}
}
{
\begin{align*}
\widehat{x_i}
&=\frac{1}{k}\cdot\langle G_i,Gx\rangle
=\frac{1}{k}\sum_{a=1}^k G_{a,i}\left(\sum_{j=1}^n G_{a,j}x_j\right)
=\frac{1}{k}\sum_{a=1}^k\left(G_{a,i}^2x_i+\sum_{j\ne i}G_{a,i}G_{a,j}x_j\right).
\end{align*}
}
Since $\Ex{G_{a,i}^2}=1$ and $\Ex{G_{a,i}G_{a,j}}=0$ for $i\ne j$, we have
\focsarxiv{
\begin{align*}
\Ex{\widehat{x_i}}
&=\frac{1}{k}\sum_{a=1}^k\Ex{G_{a,i}^2x_i+\sum_{j\ne i}G_{a,i}G_{a,j}x_j}\\
&=\frac{1}{k}\cdot k\cdot x_i
=x_i.
\end{align*}
}
{
\begin{align*}
\Ex{\widehat{x_i}}
&=\frac{1}{k}\sum_{a=1}^k\Ex{G_{a,i}^2x_i+\sum_{j\ne i}G_{a,i}G_{a,j}x_j}
=\frac{1}{k}\cdot k\cdot x_i
=x_i.
\end{align*}
}
Next we compute the variance. 
Let $Y_a=G_{a,i}\sum_{j}G_{a,j}x_j$. 
The random variables $Y_a$ are i.i.d., and hence
\begin{align*}
\Var[\widehat{x_i}]
&=\frac{1}{k^2}\sum_{a=1}^k\Var[Y_a]
=\frac{1}{k}\Var[Y_1].
\end{align*}
We expand
\focsarxiv{
\begin{align*}
\Ex{Y_1^2}
&=\Ex{G_i^2\left(\sum_j G_jx_j\right)^2}\\
&=\sum_{j,\ell}x_j\cdot x_\ell\cdot\Ex{G_i^2G_jG_\ell}.
\end{align*}
}
{
\begin{align*}
\Ex{Y_1^2}
&=\Ex{G_i^2\left(\sum_j G_jx_j\right)^2}
=\sum_{j,\ell}x_j\cdot x_\ell\cdot\Ex{G_i^2G_jG_\ell}.
\end{align*}
}
Using the Gaussian moment identities $\Ex{G_i^4}=3$ and $\Ex{G_i^2G_j^2}=1$ for $i\ne j$, we have
\begin{align*}
\Ex{Y_1^2}=3x_i^2+\sum_{j\ne i}x_j^2=2x_i^2+\|x\|_2^2.
\end{align*}
Therefore,
\focsarxiv{
\begin{align*}
\Var[Y_1]
&=\Ex{Y_1^2}-(\Ex{Y_1})^2\\
&=(2x_i^2+\|x\|_2^2)-x_i^2\\
&=x_i^2+\|x\|_2^2\le2\|x\|_2^2,
\end{align*}
}
{
\begin{align*}
\Var[Y_1]
&=\Ex{Y_1^2}-(\Ex{Y_1})^2
=(2x_i^2+\|x\|_2^2)-x_i^2
=x_i^2+\|x\|_2^2\le2\|x\|_2^2,
\end{align*}
}
and hence
\begin{align*}
\Var[\widehat{x_i}]\le\frac{2}{k}\|x\|_2^2.
\end{align*}
Then by Chebyshev's inequality,
\begin{align*}
\PPr{\left\lvert\widehat{x_i}-x_i\right\rvert\ge t}
&\le\frac{\Var[\widehat{x_i}]}{t^2}
\le\frac{2\|x\|_2^2}{kt^2}.
\end{align*}
Setting $t=\frac{4}{\sqrt{k}}\|x\|_2$ gives
\begin{align*}
\PPr{\left\lvert\widehat{x_i}-x_i\right\rvert\ge\frac{4}{\sqrt{k}}\|x\|_2}
\le\frac{1}{8},
\end{align*}
and thus with probability at least $\frac{3}{4}$,
\[
\left\lvert\widehat{x_i}-x_i\right\rvert\le\frac{4}{\sqrt{k}}\cdot\|x\|_2,
\]
as claimed.
\end{proof}

By repeating $\O{\log n}$ times and taking the median, \Cref{lem:dense:countsketch} translates to the following guarantees of dense $\CountSketch$:
\begin{theorem}
\label{thm:dense:countsketch}
Dense $\CountSketch$ outputs estimates $\widehat{x_1},\ldots,\widehat{x_n}$ such that with probability $1-\frac{1}{\poly(n)}$, for all $i\in[n]$, we have 
\[\left\lvert\widehat{x_i}-x_i\right\rvert\le\frac{4}{\sqrt{k}}\cdot\|x\|_2.\]
The algorithm uses $\O{k\log^2 n}$ bits of space. 
\end{theorem}

\noindent
Next, we recall the guarantees of the well-known $\AMS$ algorithm for $F_2$-approximation:
\begin{theorem}
\cite{AlonMS99}
\label{thm:ams}
Given a vector $x$ defined through a turnstile stream, there exists a one-pass algorithm $\AMS$ that outputs a $(1+\eps)$-approximation to $\|x\|_2$ with high probability, using $\O{\frac{1}{\eps^2}\log^2 n}$ bits of space. 
\end{theorem}
In particular, the $\AMS$ algorithm simply takes an inner product of scaled normal random variables with the input frequency vector and thus the corresponding post-processing process can be recovered from applying the dense $\CountSketch$ linear sketch.

\paragraph{Exponential random variables.}
We next recall properties of exponential random variables. 
We first define the probability density function of an exponential random variable. 
\begin{definition}[Exponential random variable]  
An exponential random variable $X$ with scale parameter $\lambda>0$ has the probability density function  
\[p(x) = \lambda e^{-\lambda x}.\]  
When $\lambda = 1$, $X$ is called a standard exponential random variable.  
\end{definition}
We next recall the following probability saying that the among a vector whose entries are scaled by inverse exponentials, the largest entry in the vector (in magnitude) is a heavy-hitter of the resulting vector with high probability. 
\begin{lemma}
\cite{EsfandiariKMWZ24}
\label{lem:max:heavy}
Let $p>0$ be an arbitrary constant, $e_1,\ldots,e_n$ be independent standard exponential random variables, let $x_1,\ldots,x_n\ge 0$, and let $c>0$ be a fixed constant.  
Then there exists a constant $C_{\ref{lem:max:heavy}}$ such that 
\begin{align*}
\PPr{\frac{\max_{i\in[n]}x_i^p/\be_i}{\sum_{i=1}^n x_i^p/\be_i}\ge\frac{1}{C_{\ref{lem:max:heavy}}\log^2 n}} \ge1-\frac{1}{n^c}.
\end{align*}
\end{lemma}
We also recall the following property that shows concentration of the order statistics of a vector whose entries are scaled by inverse exponentials. 
\begin{lemma}[Lemma 2 in \cite{JayaramW18}]
\label{lem:hidden:exps}
Given a vector $x\in\mathbb{R}^n$, let $z_{i,j}=\frac{x_i}{e_{i,j}^{1/p}}$ for independent exponential random variables $e_{i,j}$, where $j\in[n^c]$. 
Let $N=n^{c+1}$ and let $z\in\mathbb{R}^{N}$ be the resulting flattened vector. 
Similarly, let $F\in\mathbb{R}^{N}$ be the duplicated vector where each coordinate of $x$ is repeated $n^c$ times. 
Then for every integer $k\in[1,N-n^{9\zeta/10})$, we have with probability $1-\O{e^{-n^{c/3}}}$,
\[\left\lvert z_{D(k)}\right\rvert=\left[\left(1\pm\O{n^{-\zeta/10}}\right)\cdot\sum_{\tau=1}^k\frac{E_\tau}{\Ex{\sum_{j=\tau}^N|F_{D(j)}|^p}}\right]^{-1/p}.\]
\end{lemma}

\section{Perfect \texorpdfstring{$L_p$}{Lp} Sampler}
We describe our perfect $L_p$ sampler for $p\in(0,2)$. 
For each index $i\in[n]$, the algorithm in \algref{alg:perfect:lp:sample} generates a list $\calH_i$ of the top $\tau = \O{\log n}$ asymptotic order statistics of $k$ reciprocals of powers of exponential random variables $\frac{1}{k^{2/p}}\frac{1}{e_{i,j}^{2/p}}$ for $j\in[k]$, as $k\to\infty$. 
This limits the search space for the possible values of the maximum index to $\O{n \log n}$ values so that as a result, any subsequent frequency estimation algorithm only needs guarantees on $\O{n \log n}$ values. 
Moreover, by explicitly having columns in our sketch matrix for the top $\O{\log n}$ order statistics of the scalings for each of the original $n$ coordinates, we can effectively track the maximum term. 
As a result, we can upper and lower bound the parameter $R$, defined as the ratio of the maximum term $X_1$ to the sum $\sum_{j>1}X_j$ of all others, by $\polylog(n)$ and $\frac{1}{\polylog(n)}$, which is the key mechanism for achieving a polylogarithmic worst-case update time. 

Our algorithm then computes an approximation to $\widetilde{\sigma_i^2}$, a sum of transformed exponential random variables, conditioned on $\calH_i$. 
We then generate a Gaussian random variable $N_i\sim \calN(0,\widetilde{\sigma_i^2})$, which corresponds to the distribution for each bucket in the dense $\CountSketch$ data structure.

\begin{algorithm}[!htb]
\caption{Perfect $L_p$ sampler, $p<2$}
\alglab{alg:perfect:lp:sample}
\begin{algorithmic}[1]
\Require{Frequency vector $x\in\mathbb{R}^n$}
\Ensure{Sample from $i\in[n]$}
\State{$r=\O{\log n}$, $\tau=\O{\log n}$}
\For{each $i\in[n]$}
\State{Generate $\calH_i$ to be the limiting distribution as $k\to\infty$ of the largest $\tau$ values of $\left\{\frac{1}{k^{1/p}}\cdot\frac{1}{e_j^{1/p}}\right\}_{j\in[n^c]}$. Note that although $n$ is growing, it is fixed for the purposes of taking the limit, which is over $k$.}
\Comment{Exponential random variables $e_j$}
\State{Simulate $\widetilde{\sigma_i^2}=\lim_{k\to\infty}\sum_{j\in[k-\tau]}\frac{1}{k^{2/p}}\frac{1}{e_j^{2/p}}$ conditioned on $\calH_i$}
\State{Draw $N_i\sim\calN(0,\widetilde{\sigma_i^2})$ for each row}
\EndFor
\For{each stream update $s_t\in[n]$}
\State{Update $\{\calH_{s_t}\}$ and $\{N_{s_t}\}$ for $j$-th row in dense $\CountSketch$ with $r$ rows}
\EndFor
\State{At the end of the stream, recover estimates of $\calH_i$ for all $i\in[n]$}
\State{Run a statistical test, c.f., \Cref{fig:stat:test}, and upon success, return the largest estimate}
\end{algorithmic}
\end{algorithm}

Now, as new elements $s_t$ arrive in the data stream, the algorithm updates $N_{s_t}$ and the set of $\calH_{s_t}$ values. 
These updates are applied to the $j$-th row of a dense $\CountSketch$ data structure with $r=\O{\log n}$ rows, which ensures correctness over all $\O{n\log n}$ values in $\cup_{i\in[n]}\calH_i$. 
Once the stream processing is complete, the algorithm recovers estimates of $\calH_i$. 
Finally, the algorithm runs a statistical test, c.f., \Cref{fig:stat:test}, and outputs the index corresponding to the largest estimate of $\calH_i$ if the statistical test passes. 

We now describe the statistical test in \Cref{fig:stat:test} and corresponding intuition. 
We first construct a vector $z\in\mathbb{R}^{n'}$, which consists of the values $\{\calH_i\}$ and $\{N_i\}$ for all $i\in[n]$, so that $n'=\O{n\log n}$.  
Note that this is a compact representation of the previous vector $z\in\mathbb{R}^N$ because it collects the remaining duplicate coordinates outside of $\calH_i$ into $N_i$. 
We then compute a $2$-approximation $Z$ of $\|z\|_2$ and generate a fixed constant $\eps>0$ as well as a uniform random variable $\mu$ drawn from $[0.99,1.01]$. 
The purpose of the uniform random variable is to wash away the dependencies on the maximum index, along with the infinite duplication. 
Now, the statistical test checks whether $\widehat{z_{D(2)}}$ is smaller than $50\mu\eps Z$ or whether the estimated difference $\widehat{z_{D(1)}}-\widehat{z_{D(2)}}$ is smaller than $100\mu\eps Z$. 
If either condition holds, the test fails, and the procedure returns $\FAIL$. 
Otherwise, the test is considered successful and the algorithm outputs $D(1)$. 

\begin{figure*}[!htb]
\begin{mdframed}
\begin{enumerate}
\item
Let $z\in\mathbb{R}^{n'}$ be the vector with entries $\{\calH_i,N_i\}$, where $n'=\O{n\log n}$
\item 
Let $\widehat{z_{D(1)}}$ and $\widehat{z_{D(2)}}$ be the estimates of the two largest coordinates in $\cup_{i=1}^n\calH_i$ by the dense $\CountSketch$ with $R$ rows
\item
Let $Z$ be a $2$-approximation of $\|z\|_2$ from the dense $\CountSketch$
\item
Let $\eps$ be a fixed parameter and $\mu\sim\left[0.99,1.01\right]$ uniformly at random
\item
If $\widehat{z_{D(1)}}-\widehat{z_{D(2)}}<100\mu\eps Z$ or $\widehat{z_{D(2)}}<50\mu\eps Z$, return $\FAIL$
\item
Otherwise, the statistical test succeeds
\end{enumerate}
\end{mdframed}
\caption{Statistical test for perfect $L_p$ sampler in \algref{alg:perfect:lp:sample}.}
\label{fig:stat:test}
\end{figure*}

The following lemma states that adding a small perturbation $Y$ to a random variable $X$ minimally affects the probability of landing in a given region $I$, even when conditioning on an event $E$. 
The difference is at most $\O{\eps d M}$, ensuring that small noise does not significantly alter probability estimates.  
\begin{lemma}
\label{lem:interval:condition}
\cite{JayaramW18}
Let $\eps>0$ and let $X,Y\in\mathbb{R}^d$ be random variables such that $|Y|_\infty\le\eps$. 
Suppose $X$ is independent of some event $E$ and let $M>0$ be a parameter so that for every $i\in[d]$ and every $a<b$, we have $\PPr{a\le X_i\le b}\le M(b-a)$. 
Let $Z:=X+Y$ and $I=I_1\times I_2\times\cdots\times I_d\subset\mathbb{R}^d$, where each (possibly unbounded) interval $I_j=[a_j,b_j]\subset\mathbb{R}$. 
Then 
\[\PPr{Z\in I\,\mid\,E}=\PPr{Z\in I}+\O{\eps dM}.\]
\end{lemma}
We now show for any sufficiently large number of duplications, the probability of passing the statistical test remains nearly unchanged when conditioning on any possible $D(1)$, differing by at most $n^{-\O{c}}$. 
Specifically, for any desired total variation distance $\gamma=\frac{1}{n^C}$, there exists $N(\gamma)$ such that for any $n^\zeta>N(\gamma)$, the failure condition incurs at most an additive $\gamma$ distortion in the resulting sampling distribution. 
Here, we remark that $\zeta$ need not be a constant. 

This statement is crucial for showing the correctness of perfect sampling, as if the failure probability depended heavily on the index that realizes the max, this could severely distort the sampling probabilities because it could be more likely to reject certain coordinates $i\in[n]$. 
On the other hand, the proof follows almost the exact same structure as the analysis for the statistical test in \cite{JayaramW18}, with the main difference being that the failure condition only relies on the estimated maximum value, rather than the max, the second max, and their gap, as in \cite{JayaramW18}. 
\begin{lemma}
For any constant $p\in(0,2]$ and $\nu=\Theta(n^{-\zeta/20})$, we have $\PPr{\neg\FAIL\,\mid\,D(1)}=\PPr{\neg\FAIL}+n^{-\O{c}}$ for all possible $D(1)\in[N]$. 
\end{lemma}
\begin{proof}
Recall that $F\in\mathbb{R}^{N}$ is the duplicated vector where each coordinate of $x$ is repeated $n^c$ times. 
Let $\calE_1$ be the event that for every integer $k\in[1,N-n^{9\zeta/10})$, we have 
\[|z_{D(k)}|=U_{D(k)}^{1/p}\left(1\pm\O{\frac{1}{p}n^{-\zeta/10}}\right)^{1/p},\]
for a quantity $U_{D(k)}=\left(\sum_{\tau=1}^k\frac{E_\tau}{\Ex{\sum_{j=\tau}^N|F_{D(j)}|}}\right)^{-1}$ is independent of the anti-rank vector $D$. 
Then by \Cref{lem:hidden:exps}, we have $\PPr{\calE_1}=1-\O{e^{-n^{\zeta/3}}}$. 
Observe that conditioned on $\calE_1$, we can write $|z_{D(k)}|=U_{D(k)}^{1/p}(1\pm\O{\nu})$ for $\nu\ge n^{-\zeta/10}$, and more specifically, $|z_{D(k)}|=U_{D(k)}^{1/p}+U_{D(k)}^{1/p}\cdot V_{D(k)}$, for some random variable $V_{D(k)}$ not necessarily independent of $D(k)$, such that $|V_{D(k)}|=\O{\nu}$. 
Observe that for each column $G_i$, we have $A_i:=\langle G_i,x\rangle=\sum_{j\in[n^{c+1}]} g_{i,j}z_j$, where $g_{i,j}$ are independent normal random variables. 
Let $B=\{k\,\mid\,k\le N-n^{9\zeta/10}\}$ and $S=\{k\,\mid\, n^c\ge k>N-n^{9\zeta/10}\}$. 
For each $k\in[n^{c+1}]$, let $\sigma_k$ denote the sign of $z_k$, i.e., $\sigma_k=\sign(z_k)$. 
Then we have
\focsarxiv{
\begin{align*}
A_i=\sum_{k\in B} g_{i,k}\sigma_k U_{D(k)}^{1/p}&+\sum_{k\in B} g_{i,k}\sigma_k U_{D(k)}^{1/p}V_{D(k)}\\&+\sum_{k\in S} g_{i,k}z_k.
\end{align*}
}
{
\[A_i=\sum_{k\in B} g_{i,k}\sigma_k U_{D(k)}^{1/p}+\sum_{k\in B} g_{i,k}\sigma_k U_{D(k)}^{1/p}V_{D(k)}+\sum_{k\in S} g_{i,k}z_k.\]
}
\begin{claim}
\label{claim:small:terms}
For all $i$, we have
\focsarxiv{
\begin{align*}
\left\lvert\sum_{k\in B} g_{i,k}\sigma_k U_{D(k)}^{1/p}V_{D(k)}\right\rvert&+\left\lvert\sum_{k\in S} g_{i,k}z_k\right\rvert\\
&=\O{\sqrt{\log n}\cdot\nu\cdot\|z\|_2},
\end{align*}
}
{
\begin{align*}
\left\lvert\sum_{k\in B} g_{i,k}\sigma_k U_{D(k)}^{1/p}V_{D(k)}\right\rvert+\left\lvert\sum_{k\in S} g_{i,k}z_k\right\rvert
=\O{\sqrt{\log n}\cdot\nu\cdot\|z\|_2},
\end{align*}
}
with probability $1-\O{\log^2(n) n^{-c}}$. 
\end{claim}
\begin{proof}
By applying the Marcinkiewicz–Zygmund inequality, c.f., \Cref{thm:khintchine}, to a sufficiently high moment, 
\[\left\lvert\sum_{k\in S} g_{i,k}z_k\right\rvert=\O{\sqrt{\log n}\cdot\|z_S\|_2}\]
with probability $1-n^{-c}$, where $z_S$ denotes the restriction of the vector $z$ to the coordinates in $S$.  
Since $S$ consists of the smallest $n^{9\zeta/10}$ values of $|z_i|$, then we have 
\[\|z_S\|_2^2\le\frac{n^{9\zeta/10}}{N}\cdot\|z\|_2^2,\]
so that $\left\lvert\sum_{k\in S} g_{i,k}g_{i,k}z_k\right\rvert=\O{\sqrt{\log n}\cdot n^{-\zeta/20}\cdot\|z\|_2}$. 
Moreover, for $k\le N-n^{9\zeta/10}$, we have $|z_{D(k)}|<2U_{D(k)}^{1/p}$ and $|V_{D(k)}|=\O{\nu}$. 
Hence, by the Marcinkiewicz–Zygmund inequality, c.f., \Cref{thm:khintchine}, we have that with probability $1-n^{-c}$,
\[\left\lvert\sum_{k\in B} g_{i,k}z_k\right\rvert=\O{\sqrt{\log n}\cdot\nu\cdot\|z\|_2}.\]
Taking a union bound over all $\O{\log^2 n}$ indices $i\in[R]$ and scaling $c$ sufficiently, the desired result follows.
\end{proof}
Let $\calE_2$ denote the event that \Cref{claim:small:terms} holds. 
Conditioned on $\calE_2$, for each $i$, we can write
\[A_i=\sum_{k\in B} g_{i,k}\sigma_k U_{D(k)}^{1/p}+\calV_i,\]
for some random variable $\calV_i$ such that $|\calV_i|=\O{\sqrt{\log n}\cdot\nu\cdot\|z\|_2}$. 
Now, let $U_i:=\sum_{k\in B} g_{i,k}\sigma_k U_{D(k)}^{1/p}$, so that the estimate of $|z_{D(k)}|$ is $y_{D(k)}=\frac{1}{r}\sum_{i\in[r]}U_i g_{i,D(k)}$. 
We use the following claim from \cite{JayaramW18}, whose proof is the same for our case with the dense $\CountSketch$, through separately handling the ``large'' and ``small'' coordinates. 
\begin{claim}
\label{claim:ams:fail}
\cite{JayaramW18}
With probability $1-\O{n^{-c}}$, the estimate $Z$ of $\|z\|_2$ can be decomposed as $Z=U^*_Z+\calV_Z$, where $U^*_Z=\frac{5}{4}\median_i\left(\left\lvert\sum_{k\in B} g_{i,k}\sigma_k U_{D(k)}^{1/p}\right\rvert\right)$ is independent of the anti-rank vector $D$ and $|\calV_Z|=\O{\sqrt{\log(n)}\cdot\nu\cdot\|z\|_2}$. 
\end{claim}
Let $\calE_3$ be the event corresponding to \Cref{claim:ams:fail}, so that $\PPr{\calE_3}\ge 1-\O{n^{-c}}$ by \Cref{claim:ams:fail}.  

We define the deterministic function $\Lambda(x,v)=(\beta_1,\beta_2)$, where for an input vector $x$ and a scalar $v$, where $\beta_1=x_{D(1)}-x_{D(2)}-100\eps v$ and $\beta_2=x_{D(2)}-50\eps v$, so that $\Lambda(x,v)\succeq\vec{\mathbf{0}}$ coordinate-wise if and only if $\neg\FAIL$. 
We use the following claim from \cite{JayaramW18}, which we show for completeness, due to the change in the $\CountSketch$ data structure. 
\begin{claim}
\label{claim:decmopose:fail}
\cite{JayaramW18}
Conditioned on $\calE_1,\calE_2,\calE_3$, we have $\Lambda(y,\mu Z)=\Lambda(\vec{U^*},\mu U^*_Z)+\overline{V}$, where $\Lambda(\vec{U^*},\mu U^*_Z)$ is independent of $z_{D(1)}$ and $\|\overline{V}\|_\infty=\O{\sqrt{\log(n)}\cdot\nu\cdot\|z\|_2}$. 
\end{claim}
\begin{proof}
Conditioned on $\calE_1,\calE_2,\calE_3$, we have already shown that $|\calV_Z|=\O{\sqrt{\log(n)}\cdot\nu\cdot\|z\|_2}$ and $|\calV^*_{D(k)}|=\O{\sqrt{\log(n)}\cdot\nu\cdot\|z\|_2}$ for all $k\in[n^c]$. 
Since $y=\vec{U^*}+\vec{\calV^*}$ for $\vec{U}^*_{D(k)}=U^*_{D(k)}$ and $\vec{\calV}^*_{D(k)} = V^*_{D(k)}$, then it follows that $\vec{\calV}^*$ can affect the two largest coordinates of $y$ by at most $|\vec{\mathcal{V}}^|_\infty =\O{\sqrt{\log(n)}\nu|z|_2}$. 
Likewise, $|\calV_Z|$ can only affect the value of $Z$ by at most $\O{\sqrt{\log(n)}\nu|z|_2}$, which concludes the proof of the decomposition. 

It remains to show the independence claim. 
To that end, observe that $\Lambda(\vec{U^*},\mu U^*_Z)$ is a deterministic function of the hidden exponentials $E_1,\ldots,E_N$, the random Gaussians $g_{i,j}$, and the uniform random variable $\mu$. 
The joint distribution of these quantities remains marginally independent of the anti-rank vector $D$, thus completing the claim.
\end{proof}
Given \Cref{claim:decmopose:fail}, it remains to show that $\Lambda(\vec{U^*},\mu U^*_Z)$ is anti-concentrated. 
To that end, let $\calE_4$ denote the event that $\frac{1}{n^{\zeta/40}}\le g_{i,j}\le n$ for all indices $j\in\cup_{i=1}^n\calH_i$. 
From the probability density function of normal random variables, we have $\PPr{\calE_4}\ge1-\O{n^{-\zeta/40+2}}$. 
Then for any interval $I$, we have
\focsarxiv{
\begin{align*}
\PPr{\Lambda_1(\vec{U^*},\mu U^*_Z)\in I}&=\PPr{\mu\in I'/(100\eps U^*_Z)}\\
&=\O{|I|/(\eps U^*_Z)},\\
\PPr{\Lambda_2(\vec{U^*},\mu U^*_Z)\in I}&=\PPr{\mu\in I''/(50\eps U^*_Z)}\\
&=\O{|I|/(\eps U^*_Z)},
\end{align*}
}
{
\begin{align*}
\PPr{\Lambda_1(\vec{U^*},\mu U^*_Z)\in I}&=\PPr{\mu\in I'/(100\eps U^*_Z)}=\O{|I|/(\eps U^*_Z)},\\
\PPr{\Lambda_2(\vec{U^*},\mu U^*_Z)\in I}&=\PPr{\mu\in I''/(50\eps U^*_Z)}=\O{|I|/(\eps U^*_Z)},
\end{align*}
}
where $I$ denotes the size of an interval $I$, and $I'$ and $I''$ are the results of shifting the interval $I$ by term that are independent of $\mu$ and then scaling the interval $I$ by terms that have magnitude in $\left[\frac{1}{n^{\zeta/40}},n\right]$ and are independent of $\mu$. 
Let $\calE_5$ be the event that $U^*_Z$ is a $2$-approximation of $\|z\|_2$, so that $\frac{1}{2}\|z\|_2\le U^*_Z\le 2\|z\|_2$ and $\PPr{\calE_5}\ge 1-\O{n^{-c}}$. 
Then we have
\begin{align*}
\PPr{\Lambda(\vec{U^*},\mu U^*_Z)\in I}&=\O{n^{\zeta/40+1}|I|/(\eps\|z\|_2)},
\end{align*}
By \Cref{lem:interval:condition}, we have that conditioned on $\calE_1\cap\calE_2\cap\calE_3\cap\calE_4\cap\calE_5$,
\focsarxiv{
\begin{align*}
\mathbf{Pr}[&\Lambda(\vec{U^*},\mu U^*_Z)\ge 0\,\mid\,D(1)]\\
&=\PPr{\Lambda(\vec{U^*},\mu U^*_Z)\ge 0}\pm\O{\nu\cdot n^{\zeta/40+1}\log n}.
\end{align*}
}
{
\[\PPr{\Lambda(\vec{U^*},\mu U^*_Z)\ge 0\,\mid\,D(1)}=\PPr{\Lambda(\vec{U^*},\mu U^*_Z)\ge 0}\pm\O{\nu\cdot n^{\zeta/40+1}\log n}.\]
}
Since $\calE_1\cap\calE_2\cap\calE_3\cap\calE_4\cap\calE_5$ holds with probability $1-\O{n^{-\zeta/40+2}}$ and we can pick $\nu=\O{n^{-\zeta/20}}$, then it suffices to scale $c$ accordingly. 
\end{proof}
To lower bound the probability that the $L_p$ sampler will fail due to the statistical test, we first recall the following statement relating the $L_2$ norm of $z$ with the $L_p$ norm of the duplicated vector. 
\begin{lemma}[Lemma 13 in \cite{JayaramW18}]
\label{lem:mass:2p}
With probability at least $0.89$, we have that $\|z\|_2<\eps\cdot n^{c/p}\cdot\O{\|x\|_p}$, where $\eps=\Theta(1)$ for $p<2$. 
\end{lemma}
This bound shows that the element $z_{D(1)}$ largest in magnitude is an $L_2$ heavy-hitter of $z$ with good probability. 
Thus, we can lower bound the probability that the $L_p$ sampler will fail due to the statistical test as follows. 
\begin{lemma}
There exist universal parameters $\eta=\Theta(1)$ for $p<2$ and $\eps=\frac{1}{\Theta(\sqrt{\log n})}$ such that the statistical test succeeds with probability at least $0.75$. 
\end{lemma}
\begin{proof}
Recall that $F\in\mathbb{R}^{N}$ is the duplicated vector where each coordinate of $x$ is repeated $n^c$ times, so that $\|F\|_p=n^{c/p}\cdot\|x\|_p$
Let $\calE_1$ be the event that $\|z\|_2<\eta\cdot\O{\|F\|_p}$, where $\eta=\Theta(1)$ for $p<2$, so that $\PPr{\calE_1}\ge 0.89$ by \Cref{lem:mass:2p}. 
Let $\calE_2$ be the event that the estimate $Z$ satisfies $\frac{1}{2}\|z\|_2\le Z\le 2\|z\|_2$ so that $\calE_2$ occurs with high probability by \Cref{thm:ams}. 
Let $\calE_3$ be the event that dense $\CountSketch$ outputs an estimate vector $y$ such that $\max_{i\in[n]}\left\lvert y_i-|z_i|\right\rvert\le 2\eps\|z\|_2$, so that $\calE_3$ occurs with high probability by \Cref{thm:dense:countsketch}. 

We have $|z_{D(1)}|=\frac{\|F\|_p}{E_1^{1/p}}$ and $|z_{D(2)}|=\frac{\|F\|_p}{(E_1+E_2(1\pm n^{-c+1}))^{1/p}}$ for independent exponential random variables $E_1$ and $E_2$. 
From the probability density function of exponential random variables, we have that with probability at least $\frac{7}{8}$, we simultaneously have $|z_{D(1)}|=\Theta(\|F\|_p)$ and $|z_{D(2)}|=\Theta(\|F\|_p)$, and $|z_{D(1)}|-|z_{D(2)}|=\Theta(\|F\|_p)$, with sufficiently scaled constants. 
Hence, conditioned on $\calE_1$, $\calE_2$, and $\calE_3$, we have $y_{D(1)}$ and $y_{D(2)}$ have additive error $\O{\|F\|_p}$ to the values of $z_{D(1)}$ and $z_{D(2)}$ with a sufficiently small constant. 
Therefore, we have that $y_{D(1)}-y_{D(2)}>100\eps\mu Z$ and $y_{D(2)}>50\eps\mu Z$ with probability at least $0.75$. 
\end{proof}

Finally, we remark that as stated, the Gaussian random variables, exponential random variables, and the corresponding simulations are continuous random variables, which cannot be stored in the finite precision necessary for streaming algorithms. 
To that end, we use a standard procedure of truncating the binary representation of each continuous random variable after $\O{\log n}$ bits for a sufficiently large constant. 
This results in at most an additive error of $\frac{1}{\poly(n)}$ error for each cell of the dense $\CountSketch$ matrix, which can be absorbed into the adversarial error $V_{i,j}$, where $|V_{i,j}| = \O{\nu}\cdot|A_i|$, which already occurs in each of the cells, as previously shown. 
Hence, each random variable requires $\O{\log n}$ bits to encode. 

\subsection{Justification of Simulation Oracle via Poisson Approximation}
To achieve $\polylog(n)$ update time, our algorithm will require a new simulation oracle that bypasses the explicit duplication of the $k=n^c$ exponential random variables for each coordinate $i\in[n]$. 
To do this, our algorithm needs to simulate, for each $i \in [n]$, the joint distribution of the largest values and the aggregate norm of the rest. 
Specifically, let $\{V_{(1)}, \ldots, V_{(k)}\}$ be the order statistics of $\{ (e_{i,j})^{-1/p} \}_{j=1}^k$. 
For $\tau = \Theta(\log n)$, we are interested in the joint statistics of the ``head'' and ``tail'', compactly represented as follows:
\[ 
\mathcal{S}_{\text{true}}^{(k)} = \left( V_{(1)}, \ldots, V_{(\tau)}, \sum_{j=\tau+1}^k V_{(j)}^2 \right).
\]
We observe that we can efficiently simulate a random vector $\mathcal{S}_{\text{sim}}^{(\infty)}$ whose joint distribution is statistically indistinguishable from that of $\mathcal{S}_{\text{true}}^{(k)}$ after normalization, using $\polylog(n)$ time. 

The algorithm then proceeds as described in \algref{alg:perfect:lp:sample}, simulating the top $\tau$ values $\mathcal{H}_i = \{\mathcal{V}_{(1)}, \ldots, \mathcal{V}_{(\tau)}\}$ for each $i\in[n]$, by sampling from the known \emph{limiting} distribution of the top $\tau$ order statistics obtained in the limit of large $k$. 
We then shall use an oracle from \Cref{sec:calculation} to simulate the sum-of-squares of the remaining tail, $\widetilde{\sigma_i^2}$, conditioned on the value of $\mathcal{V}_{(\tau)}$. 
In this section, we justify the correctness of the limiting distribution. 

We first recall the definitions of Bernoulli and Poisson point processes.
\begin{definition}[Bernoulli point process]
A \emph{Bernoulli point process} on a space $\mathcal{X}$ with parameter function $p: \mathcal{X} \to [0,1]$ is a random finite subset $\Phi \subseteq \mathcal{X}$ obtained by including each point $x \in \mathcal{X}$ independently with probability $p(x)$.  
\end{definition}

\begin{definition}[Poisson Point Process]
A \emph{Poisson point process} on a space $\mathcal{X}$ with intensity measure $\Lambda$ is a random subset $\Phi \subseteq \mathcal{X}$ such that:
\begin{enumerate}
\item 
The total number of points $N = |\Phi|$ follows the Poisson distribution $\Poi(\Lambda(\mathcal{X}))$.
\item 
Conditioned on $N=k$, the points of $\Phi$ are independent and identically distributed according to the normalized measure $\Lambda / \Lambda(\mathcal{X})$.
\end{enumerate}
\end{definition}
An important quantifier for a Poisson point process is its intensity measure and the related notion of the intensity density. 
\begin{definition}[Intensity Measure and Intensity Density]
Let $\Phi$ be a Poisson point process on a space $\mathcal{X}$.
The \emph{intensity measure} $\Lambda$ of $\Phi$ is defined by $\Lambda(A) = \Ex{\Phi(A)}$, giving the expected number of points of $\Phi$ in any region $A \subseteq \mathcal{X}$.

If $\Lambda$ is absolutely continuous with respect to a reference measure $\mu$, then the \emph{intensity density} is the derivative $\lambda(x) = \frac{d\Lambda}{d\mu}(x)$, which specifies the expected number of points per unit $\mu$-mass around $x$.
\end{definition}

We next recall the following formulation of Le Cam's Theorem:
\begin{theorem}[Le Cam's Theorem]
\label{thm:le:cam}
\cite{le1960approximation,lecam1965distribution,steele1994cam}
Let $X_1,\ldots,X_n$ be independent Bernoulli random variables with
$\PPr{X_i=1}=p_i$ and $\PPr{X_i=0}=1-p_i$. 
Let
\[S=\sum_{i=1}^n X_i,\qquad \lambda=\Ex{S}=\sum_{i=1}^n p_i,\]
and let $\Poi(\lambda)$ denote the Poisson distribution with mean $\lambda$. 
Then the total variation distance between the distribution of $S$ and $\Poi(\lambda)$ is at most $\sum_{i=1}^n p_i^2$. 
\end{theorem}
We show that show that the true Binomial point process and the simulated Poisson point process have a small total variation distance through a hybrid argument. 
Namely, we consider the following:
\focsarxiv{
\newline\noindent
\textbf{True Binomial point process $\calP_k$.} 
This is a BPP consisting of $k=n^C$ i.i.d. points $\{V_i^*\}_{i=1}^k$, where $V_i^* = V_i / k^{1/p}$ and $\PPr{V_i^* > y} = 1 - e^{-y^{-p}/k}$.
\newline\noindent
\textbf{Simulated Poisson point process $\calP_\infty$.}
This is a PPP with intensity measure $\lambda_\infty$ defined by its tail: $\lambda_\infty(y, \infty) = y^{-p}$.
\newline\noindent
\textbf{Intermediate Poisson point process $\calP'_k$.}
To bridge the gap, we define a ``proxy'' PPP whose intensity measure $\lambda_k$ matches the expected value of the BPP. 
Its tail is $\lambda_k(y, \infty) = k \cdot \PPr{V_i^* > y} = k(1 - e^{-y^{-p}/k})$.
\newline

}
{
\begin{description}
\item[True Binomial point process $\calP_k$.] 
This is a BPP consisting of $k=n^C$ i.i.d. points $\{V_i^*\}_{i=1}^k$, where $V_i^* = V_i / k^{1/p}$ and $\PPr{V_i^* > y} = 1 - e^{-y^{-p}/k}$.
\item[Simulated Poisson point process $\calP_\infty$.] This is a PPP with intensity measure $\lambda_\infty$ defined by its tail: $\lambda_\infty(y, \infty) = y^{-p}$.
\item[Intermediate Poisson point process $\calP'_k$.] 
To bridge the gap, we define a ``proxy'' PPP whose intensity measure $\lambda_k$ matches the expected value of the BPP. 
Its tail is $\lambda_k(y, \infty) = k \cdot \PPr{V_i^* > y} = k(1 - e^{-y^{-p}/k})$.
\end{description}
}
We remark that \algref{alg:perfect:lp:sample} performs the process $\calP_\infty$ by directly implementing the two defining statistical properties of a PPP with the specific power-law intensity $\lambda_\infty(y, \infty) = y^{-p}$, while focusing only on the interval $A=[y_0, \infty)$. 
This simulation is justified because it is a closed-form, efficient method that relies on known principles of point process theory:
\begin{enumerate}
\item \textbf{Poisson Count Law:} 
In a PPP, the number of points $N$ in a set $A$ must follow a Poisson distribution with mean $\lambda_A = \mu(A)$. 
\algref{alg:perfect:lp:sample} correctly implements this by sampling $N$ from $\Poi(y_0^{-p})$, where $y_0^{-p}$ is the expected number of points $\mathcal{P}_\infty$ has in the region $A$.
\item 
\textbf{I.I.D. Power-Law Locations:} 
Conditioned on having $N$ points, the locations of the points must be independent and identically distributed (i.i.d.) according to the normalized density function derived from the intensity $\lambda_\infty$. 
For intensity $y^{-p}$, the location density is proportional to $y^{-p-1}$. 
\algref{alg:perfect:lp:sample} can use the inverse CDF method to sample from this exact power-law density to generate the $N$ locations, $\{Y_i\}_{i=1}^N$, ensuring they are placed according to the $\calP_\infty$ distribution within $A$.
\end{enumerate}
By accurately and efficiently generating the correct count and location statistics within the interval $A$, \algref{alg:perfect:lp:sample} simulates the idealized process $\calP_\infty$. 
Thus, it remains to argue that the Poisson process $\calP_\infty$ has small total variation distance from the Binomial process $\calP_k$ through a hybrid argument. 
We now argue the first part of the hybrid argument. 
\begin{lemma}
\label{lem:tvd:first}
Let $y_0=\frac{1}{n^{C'}}$ for some constant $C'>0$ such that $\poly\left(\frac{1}{y_0}\right)\ll k$ and let $A=[y_0,\infty)$. 
Then the total variation distance between $\calP_k(A)$ and $\calP'_k(A)$ is at most $\O{\frac{1}{\sqrt{k}}}$. 
\end{lemma}
\begin{proof}
This is the standard distance between a Binomial process and a Poisson process with the same mean. 
Since $\calP_k(A)$, i.e., the process $\calP_k$ on set $A$, is a collection of $k$ Bernoulli trials, where each point lands in $A$ with probability $p_A = \Pr(V_i^* \in A) = 1 - e^{-y_0^{-p}/k}$. 
The total count $N_k(A)$ is $\text{Bin}(k, p_A)$. 
On the other hand, the process $\calP'_k(A)$ is by definition, a Poisson process with mean $\lambda_A = k \cdot p_A$.

By Le Cam's Theorem, c.f., \Cref{thm:le:cam}, the total variation distance between the counts is upper bounded by the sum of squares of the individual trial probabilities $k \cdot p_A^2$. 
By Taylor expansion,
\[p_A = 1 - e^{-y_0^{-p}/k} = \frac{y_0^{-p}}{k} + \O{k^{-2}}.\]
Thus for $\poly\left(\frac{1}{y_0}\right)\ll k$, the total variation distance is at most
\focsarxiv{
\begin{align*}
k \cdot p_A^2 &\le k \cdot \left( \frac{y_0^{-p}}{k} + \O{k^{-2}} \right)^2 \\
&= k \left( \frac{y_0^{-2p}}{k^2} + \O{k^{-3}} \right) \\
&= \frac{y_0^{-2p}}{k} + \O{k^{-2}}.
\end{align*}
}
{
\[k \cdot p_A^2 \le k \cdot \left( \frac{y_0^{-p}}{k} + \O{k^{-2}} \right)^2 = k \left( \frac{y_0^{-2p}}{k^2} + \O{k^{-3}} \right) = \frac{y_0^{-2p}}{k} + \O{k^{-2}}.\]
}

While Le Cam's Theorem only upper bounds the total variation process between the number of points that fall within $A$, this argument extends to the entire process including point locations because, conditioned on $N$ points landing in $A$, their locations are identically distributed in both the BPP and PPP models, as both are i.i.d. from the same normalized probability distribution on $A$. 
Therefore, the total variation distance is at most
\[\frac{y_0^{-2p}}{k} + \O{k^{-2}} = \O{\frac{1}{\sqrt{k}}}.\]
\end{proof}
We next argue the second part of the hybrid argument. 
\begin{lemma}
\label{lem:tvd:second}
Let $y_0=\frac{1}{n^{C'}}$ for some constant $C'>0$ such that $\poly\left(\frac{1}{y_0}\right)\ll k$ and let $A=[y_0,\infty)$. 
Then the total variation distance between $\calP'_k(A)$ and $\calP_\infty(A)$ is at most $\O{\frac{1}{\sqrt{k}}}$. 
\end{lemma}
\begin{proof}
This is the distance between two different Poisson processes, $\calP'_k$ and $\calP_\infty$. 
The total variation distance between two Poisson processes is bounded by the $L_1$ difference of their intensity measures. 
Thus, we upper bound this distance on our region $A=[y_0, \infty)$ by considering the intensity density $\mu(y)$, the negative derivative of the tail measure $\lambda(y, \infty)$, as follows:
\[\mu_\infty(y) = - \frac{d\lambda_\infty}{dy} = - \frac{d}{dy}(y^{-p}) = p y^{-p-1}.\]
\begin{align*}
\mu_k(y) = - \frac{d\lambda_k}{dy} &= - \frac{d}{dy}\left(k(1-e^{-y^{-p}/k})\right) \\
&= -k \left(-e^{-y^{-p}/k}\right) \cdot \frac{d}{dy}(-y^{-p}/k) \\
&= k e^{-y^{-p}/k} \cdot \left(-\frac{1}{k}\right) \cdot (-py^{-p-1}) \\
&= p y^{-p-1} e^{-y^{-p}/k}
\end{align*}
The $L_1$ difference, which upper bounds the total variation distance between the process, is at most
\focsarxiv{
\begin{align*}
\int_{y_0}^\infty &\left| \mu_k(y) - \mu_\infty(y) \right| dy \\
&= \int_{y_0}^\infty \left| p y^{-p-1} e^{-y^{-p}/k} - p y^{-p-1} \right| dy \\
&= \int_{y_0}^\infty p y^{-p-1} \left| e^{-y^{-p}/k} - 1 \right| dy \\
&= p \int_{y_0}^\infty y^{-p-1} (1 - e^{-y^{-p}/k}) dy,
\end{align*}
since $e^{-x} \le 1$. 
}
{
\begin{align*}
\int_{y_0}^\infty \left| \mu_k(y) - \mu_\infty(y) \right| dy &= \int_{y_0}^\infty \left| p y^{-p-1} e^{-y^{-p}/k} - p y^{-p-1} \right| dy \\
&= \int_{y_0}^\infty p y^{-p-1} \left| e^{-y^{-p}/k} - 1 \right| dy \\
&= p \int_{y_0}^\infty y^{-p-1} (1 - e^{-y^{-p}/k}) dy \quad (\text{since } e^{-x} \le 1).
\end{align*}
}
We perform the substitution $u=y^{-p}$. 
Then the integral becomes
\begin{align*}
p \int_{u=y_0^{-p}}^{0} &(1 - e^{-u/k}) \cdot \left(-\frac{1}{p} du\right)\\
&= - \int_{y_0^{-p}}^{0} (1 - e^{-u/k}) du \\
&= \int_{0}^{y_0^{-p}} (1 - e^{-u/k}) du.
\end{align*}
Using the inequality $1 - e^{-x} \le x$ for $x \ge 0$, we have
\begin{align*}
\int_{0}^{y_0^{-p}} (1 - e^{-u/k}) du \le \int_{0}^{y_0^{-p}} \frac{u}{k} du = \frac{y_0^{-2p}}{2k}.
\end{align*}
Thus, the total variation distance between $\calP'_k(A)$ and $\calP_\infty(A)$ is at most $\frac{y_0^{-2p}}{2k}=\O{\frac{1}{\sqrt{k}}}$ for $\poly\left(\frac{1}{y_0}\right)\ll k$.
\end{proof}
We now justify the correctness of the simulation process.
\begin{lemma}
\label{lem:gclt:justification}
Let $k = n^C$ for a sufficiently large constant $C > 0$. 
Let $V_j = (E_j)^{-1/p}$ for $E_j \sim \text{Exp}(1)$, and let $b_k = k^{1/p}$ be the normalization factor.
\begin{enumerate}
\item 
Let $\mathcal{P}_k = \{ V_j / b_k \}_{j=1}^k$ be the normalized Binomial point process (BPP) of the $k$ true (finite) samples.
\item 
Let $\mathcal{P}_\infty$ be the Poisson point process (PPP) on $\mathbb{R}^+$ with intensity measure $\mu(y, \infty) = y^{-p}$.
\end{enumerate}
The sampler's statistics $\mathcal{S}_{\text{true}}^{(k)}$ (normalized) are a function of $\mathcal{P}_k$. 
The simulated statistics $\mathcal{S}_{\text{sim}}^{(\infty)}$ are the same function applied to $\mathcal{P}_\infty$. 
Moreover, the total variation distance between these two underlying processes is a most $\O{\frac{1}{\sqrt{k}}}= \O{n^{-C/2}}$.
\end{lemma}
\begin{proof}
By \Cref{lem:tvd:first}, \Cref{lem:tvd:second}, and the triangle inequality, the total variation distance between these two process on the region $A=[y_0, \infty)$ is at most $\O{\frac{1}{\sqrt{k}}}$. 
Since our algorithm only uses statistics from the region $A=[y_0, \infty)$ where $y_0 = n^{-C'} > 0$ with high probability, the constant $y_0^{-2p} = (n^{-C'})^{-2p} = n^{2pC'}$ is just some polynomial in $n$. 
The error bound is $\O{y_0^{-2p}/k} = \O{n^{-C/2}}$, by choosing a sufficiently large $k=n^C$. 
This proves the total variation distance between the true finite process, which our algorithm would ideally be simulating, and the limiting process, which our algorithm actually simulates in reality, is negligibly small on the region of interest.
\end{proof}

\subsection{Achieving Fast Update Time} 
Here we collect several facts which together guarantee that our sampler runs quickly.  
The following fact shows that we will be able to take $R$ and $1/R$ to be at most $\polylog(n)$ in our sampler, where we recall $R$ is the ratio of the maximum term $X_1$ to the sum $\sum_{j>1}X_j$ of all others. 
We will apply this result to the random variables $X_i^{2/p}$.

\begin{lemma}
Let $X_1 > X_2 > \ldots > X_k$ be $k$ independent and identically distributed inverse exponential random variables arranged in decreasing order, with $k\geq \log n$. 
Then there exist constants $C,c, c_1, c_2>0$ such that
\begin{enumerate}
    \item $X_{C\log n} \geq c_1 k/\log n$
    \item $X_{C\log n} \leq c_2 k$
\end{enumerate}
with failure probability at most $n^{-c.}$
\end{lemma}

\begin{proof}

Let $E_1 < \ldots < E_k$ be the reciprocals of the $X_i$'s.

For the first point we argue that $C \log n$ of $E_i$'s are at most $\log n/k.$  The probability that a given $E_i$ is this small is $1-\exp(-\log n / k) \geq \frac{1}{2} \log n /k)$ for $k\geq \log n$.  
In any group of $k/\log n$, there is therefore at least a $1/2$ probability of having at least one exponential this small.  So the first claim follows by applying a Chernoff bound to the $\log n$ groups.

For the second point, we note that $E_i \geq \frac{1}{k}$ with probability $1 - \O{\frac{1}{k}}$. 
The probability that $C\log n$ events of probability $1/k$ occur is at most $\exp(-c\log n) = n^{-c}$ by a Chernoff bound.
\end{proof}

Our sampler will perform a binary search on the PDF to perform the sampling.  
This lemma shows that it is unlikely to sample something very large. 

\begin{lemma}
Let $X_1, \ldots, X_k$ be inverse exponentials truncated at $R^{p/2}k$.  Then 
\[
X_1^{2/p} + \ldots + X_k^{2/p} \leq C R k^{2/p} \polylog(R,n)
\]
with probability at least $1 - n^{-c}.$
\end{lemma}
\begin{proof}
We use the standard technique of partitioning the possible range of values into level sets and bounding the number of variables falling into each set.
We partition the range of possible values for $X_i$ into logarithmic level sets based on powers of $2$. 
Since $X_i$ is truncated at $R^{p/2}k$, the maximum index $J$ is defined such that $2^J\ge R^{p/2}k\ge 2^{J-1}$. 
We define the level sets are defined for $i \geq 0$ by $L_i = [2^i, 2^{i+1})$. 
The total sum can be rewritten in terms of $N_i$, the number of variables that fall into $L_i$:
\[\sum_{j=1}^k X_j^{2/p} = \sum_{i=0}^J \sum_{X_j \in L_i} X_j^{2/p}.\]
We bound the contribution of each level set by its upper endpoint:
\[\sum_{j=1}^k X_j^{2/p} \le \sum_{i=0}^J N_i \cdot (2^{i+1})^{2/p} = 2^{2/p} \sum_{i=0}^J N_i \cdot (2^i)^{2/p},\]
where $N_i = \left| \{ j \mid X_j \in L_i \} \right|$.

For an inverse exponential variable $X = Z^{-1}$ where $Z \sim \mathrm{Exp}(1)$, the probability that $X$ falls into the range $[2^i, \infty)$ is $\PPr{X \geq 2^i} = 1 - e^{-2^{-i}}=\Theta(2^{-i})$. 
The expected number of variables that fall into the $i$-th level set is $\Ex{N_i}=\Theta(k\cdot 2^{-i})$. 

We first consider the low-contribution sets where $\Ex{N_i}\ge\log n$, corresponding to $i \leq i_0 = \log k-\O{\log\log n}$. 
For these sets, we have by Chernoff bounds that $\PPr{N_i\le\frac{Ck}{2^i}}\ge 1-n^{-ci}$. 
Then the contribution from these sets is bounded by a geometric series:
\[\sum_{i=0}^{i_0} N_i (2^i)^{2/p} \leq \sum_{i=0}^{i_0}\frac{Ck}{2^i}\cdot (2^{2/p})^i.\]
Observe that $\frac{Ck}{2^i}\cdot (2^{2/p})^i$ is maximized at $i=i_0=\log k-\O{\log\log n}$ by $\O{k^{2/p}}$ and since $i_0=\O{\log n}$, then we have $\sum_{i=0}^{i_0} N_i (2^i)^{2/p} = \O{k^{2/p}\cdot\log n}$.

It remains to consider the level sets where $\Ex{N_i}<\log n$. 
These correspond to the heavy-tailed, high-value coordinates. 
We again apply Chernoff bounds to ensure that the count $N_i$ for these levels does not exceed $\O{\log n}$ with high probability. 
Specifically, the number of elements in the highest $\O{\log k}$ level sets, denoted $N_{\text{top}}$, is upper bounded by $\polylog(n)$. 
Since each value is truncated at $R^{p/2}k$ by hypothesis, then the contribution from the these terms is thus upper bounded by $(R^{p/2}k)^{2/p}\cdot\polylog(n)$. 
Combining the negligible contribution from the low-value terms with the dominating bound from the high-value terms, we conclude that with probability at least $1 - n^{-c}$,
\[\sum_{i=1}^k X_i^{2/p} \leq C R k^{2/p} \polylog(R, n).\]
\end{proof}

\section{Instantiating the Sampling Oracle}\label{sec:calculation}
In this section, we describe how to instantiate the sampling oracle, which is crucial to achieving polylogarithmic update time in expectation. 
Without such a sampling oracle, one would need to generate and maintain a polynomial number of random variables at each time due to the duplication. 
As a result, we would incur polynomial update time without such a sampling oracle. 

\subsection{Sampling Lemma}
Simulating duplications required sampling accurately from a distribution. 
In this section we construct the sampling oracle that we need.

Let $X_1, X_2, \ldots, X_k$ be $k$ i.i.d. truncated inverse exponential random variables, raised to the $2/p$-th power, and bounded by $Rk^{2/p}.$ Let $\mathcal{D}_{p,R,k}$ be the distribution of
\[
\frac{1}{k^{2/p}}(X_1 + \ldots + X_k ).
\]

We say that that an algorithm samples from a distribution $\mathcal{D}$ to $L$ bits of precision if it simulates drawing from $\mathcal{D}$ and rounding to one of the two nearest $L$-bit precision numbers.

\begin{theorem}\label{thm:calculation}
Suppose that $R = \frac{1}{\Omega(\log n)}.$ For $0 < p < 2,$ there is an algorithm that samples from $\mathcal{D}_{p, R, \infty}$ to $L$ bits of precision in $\tO{LR}$ time.
\end{theorem}

The idea is to compute the characteristic function of the limiting distribution $\mathcal{D}_{p,R,\infty}$. Once we have this characteristic function, we may apply a standard Fourier inversion formula to express the CDF at a given value $t$ as an integral.  The resulting integral analytic, and ideas from numerical analysis show that an associated quadrature rule converges exponentially fast.  Thus we have fast CDF evaluation.  Given this, we can pick a uniformly random $y \in [0,1]$ and binary search for the value $x$ where the CDF evaluates to $y.$ By the lemmas in the previous section, we may assume that $R, 1/R$ and $t$ are $\polylog n$, so it suffices to give an algorithm for estimating the CDF to within $\eps$ that requires only $\poly(R, \frac{1}{R}, t, \log\frac{1}{\eps})$ time. 
We show how to accomplish this in the following sections.

\subsection{Gamma Function Facts}
We will use the incomplete gamma functions extensively in our calculations. 
In this section we recall several standard facts, and derive consequences of them that we will apply later. 
We first collect several standard identities for the incomplete gamma functions. 
In particular, we denote the upper incomplete gamma function by
\[\Gamma(s,z)=\int^\infty_z t^{s-1}e^{-t}\,dt\]
and the lower incomplete gamma function by
\[\gamma(s,z)=\int^z_0 t^{s-1}e^{-t}\,dt.\]
\begin{proposition}
\label{prop:gamma_facts}
We have the following standard facts for the incomplete gamma functions.
\begin{enumerate}
\item $\Gamma(s,z) = \Gamma(s) - \gamma(s,z)$
\item $\gamma(s+1,z) = s\gamma(s,z) - z^s e^{-z}$
\item $\Gamma(s+1,z) = s\Gamma(s,z) + z^s e^{-z}$
\item $\frac{d}{dz} \Gamma(s,z) = -z^{s-1}e^{-z}$
\item $\frac{d}{dz} \gamma(s,z) = z^{s-1}e^{-z}$
\item For fixed $s$, the upper incomplete gamma function has the asymptotic series
\[
\Gamma(s,z) = z^{s-1}e^{-z} \sum_{k=0} \frac{\Gamma(s)}{\Gamma(s-k)} z^{-k},
\]
valid as $\abs{z} \rightarrow \infty$ in the sector $\{z : \abs{\arg{z}} \leq \frac{3\pi}{4}\}.$
\item In particular, from the first term of the asymptotic series, there is a constant $c_s$ depending only on $s$ such that
$\Gamma(s,z)\leq c_s \abs{z^{s-1} e^{-z}}$ for all $z$ in the sector.
\item There is an entire function $\gamma_*(s,z)$ such that for all $\gamma(s,z) = \Gamma(s) z^s \gamma_*(s,z).$
\end{enumerate}
\end{proposition}
Recall that a function is entire if it is analytic at every point in the complex plane. 
Now, for a complex number $z$, let $\Re(z)$ and $\Im(z)$ denote its real and imaginary parts, respectively. 
\begin{proposition}
\label{prop:h_bounded_deriv}
For $s\in(0,1]$, let $h(z) = z^{s} \gamma(-s,z)$.  
Then $h'(z)$ is bounded by a constant $C_s$ on the region $\{z: \Re(z) \geq -1\}$.
\end{proposition}

\begin{proof}
Using the identities in \Cref{prop:gamma_facts} we have
\begin{align*}
h'(z) &= s z^{s-1} \gamma(-s,z) + z^s \frac{d}{dz} \gamma(-s,z)\\
&= s z^{s-1} \gamma(-s,z) + z^s (z^{-s-1}\exp(-z))\\
&= s z^{s-1} \gamma(-s,z) + z^{-1}\exp(-z)\\
&= z^{s-1}\left(\frac{\gamma(1-s,z) + z^{-s}e^{-z}}{-s}\right) + z^{-1}e^{-z}\\
&= -z^{s-1}\gamma(1 - s, z)\\
&= -z^{s-1}(\Gamma(1-s) - \Gamma(1-s, z)).
\end{align*}
By \Cref{prop:gamma_facts}, for $\abs{z} > 1$ we have that $\abs{z}$ lies in the sector and so using the first term of the asymptotic expansion of $\Gamma(\cdot, \cdot),$ we get the following bound:
\begin{align*}
\abs{h'(z)}
&\leq \abs{z}^{s-1}(\abs{\Gamma(1-s)} + \abs{ \Gamma(1-s,z)})\\
&\leq \abs{z}^{s-1}(\abs{\Gamma(1-s)} + c_s\abs{z}^{-s}e^{-\Re(z)})\\
&= \abs{\Gamma(s)}\abs{z}^{s-1} + c_s \abs{z}^{-1}e^{-\Re(z)}.
\end{align*}
Since $s\in(0,1]$, then it follows that $|h'(z)|$ is bounded by a constant on the region $|z|>1$.

This establishes the bound on $\{z : \Re(z) > -1 \cap |z| \geq 1\}.$ Finally, note that $h'(z)$ is entire for fixed $s.$ Therefore by continuity, along with compactness of the unit disk, $h'(z)$ achieves its maximum inside the closed unit disk. Hence, it follows that $\abs{h'(z)}$ is bounded on $\{z: \Re(z) \geq -1\}.$
\end{proof}

\begin{proposition}
\label{prop:max_imaginary_part_zero}
Let $f(z) = z^s \gamma(-s,z)$ where $s \in (0,1).$  Then for fixed $x\in \mathbb{R}$, $f(x + iy)$ is maximized over $y\in \mathbb{R}$ when $y=0.$ 
\end{proposition}
\begin{proof}
The analytic continuation of the lower incomplete gamma function is given by its series expansion:
$$ \gamma(a, z) = \sum_{n=0}^\infty \frac{(-1)^n z^{a+n}}{n! (a+n)}. $$
Substituting $a=-s$, we have:
$$ \gamma(-s, z) = \sum_{n=0}^\infty \frac{(-1)^n z^{n-s}}{n! (n-s)}. $$
The function $f(z)$ is thus given by:
\focsarxiv{
\begin{align*}
f(z) &= z^s \gamma(-s, z)\\ 
&= z^s \sum_{n=0}^\infty \frac{(-1)^n z^{n-s}}{n! (n-s)} \\
&= \sum_{n=0}^\infty \frac{(-1)^n z^n}{n! (n-s)}.
\end{align*}
}
{
$$ f(z) = z^s \gamma(-s, z) = z^s \sum_{n=0}^\infty \frac{(-1)^n z^{n-s}}{n! (n-s)} = \sum_{n=0}^\infty \frac{(-1)^n z^n}{n! (n-s)}. $$
}
This representation shows that $f(z)$ is an entire function.

Now let $s\in(0,1)$. We claim that the function $f(z)$ can be represented as:
$$ f(z) = \int_0^1 t^{-s-1} (e^{-zt} - 1) dt - \frac{1}{s}. $$ To see this, let $I(z)$ denote the right-hand side. We first verify the convergence of the integral. Let $h(t, z) = e^{-zt} - 1$. Near $t=0$, $h(t, z) = -zt + O(t^2)$. The integrand is $t^{-s-1} h(t, z) = -z t^{-s} + O(t^{1-s})$. Since $s < 1$, we have $-s > -1$ and $1-s > 0$. Thus, the integral converges at $t=0$. We expand $h(t, z)$ into its Taylor series:
$$ I(z) = \int_0^1 t^{-s-1} \sum_{n=1}^\infty \frac{(-zt)^n}{n!} dt - \frac{1}{s}. $$
We justify the interchange of summation and integration using the Dominated Convergence Theorem (DCT). Let $S_N(t, z) = \sum_{n=1}^N \frac{(-zt)^n}{n!}$. We seek an integrable dominating function for $t^{-s-1} S_N(t, z)$.
We know that $|S_N(t, z)| \leq \sum_{n=1}^N \frac{|z|^n t^n}{n!} \leq e^{|z|t} - 1$.
For $t \geq 0$, we have the inequality $e^{|z|t} - 1 \leq |z|t e^{|z|t}$.
Thus,
$$ |t^{-s-1} S_N(t, z)| \leq t^{-s-1} |z|t e^{|z|t} = |z| t^{-s} e^{|z|t}. $$
Let $G(t) = |z| t^{-s} e^{|z|t}$. Since $-s > -1$, $G(t)$ is integrable on $[0, 1]$. By DCT, we can interchange the summation and integration:
\begin{align*}
I(z) &= \sum_{n=1}^\infty \frac{(-z)^n}{n!} \int_0^1 t^{n-s-1} dt - \frac{1}{s}.
\end{align*}
Since $n \geq 1$ and $s < 1$, $n-s-1 > -1$. The integral evaluates to:
$$ \int_0^1 t^{n-s-1} dt = \frac{1}{n-s}. $$
Therefore,
$$ I(z) = \sum_{n=1}^\infty \frac{(-1)^n z^n}{n! (n-s)} - \frac{1}{s}. $$
Comparing this with the series representation of $f(z)$:
$$ f(z) = \frac{1}{0-s} + \sum_{n=1}^\infty \frac{(-1)^n z^n}{n! (n-s)} = -\frac{1}{s} + \sum_{n=1}^\infty \frac{(-1)^n z^n}{n! (n-s)}. $$
We conclude that $f(z) = I(z)$.

Now set $z = x+iy$. Using the integral representation from above,
$$ f(x+iy) = \int_0^1 t^{-s-1} (e^{-(x+iy)t} - 1) dt - \frac{1}{s}. $$
We examine the real part $g(y)$:
\begin{align*}
g(y) &= \text{Re}\left(\int_0^1 t^{-s-1} (e^{-xt}e^{-iyt} - 1) dt\right) - \frac{1}{s} \\
&= \int_0^1 t^{-s-1} \text{Re}(e^{-xt}e^{-iyt} - 1) dt - \frac{1}{s} \\
&= \int_0^1 t^{-s-1} (e^{-xt}\cos(yt) - 1) dt - \frac{1}{s}.
\end{align*}
We compare $g(y)$ with $g(0)$:
$$ g(0) = \int_0^1 t^{-s-1} (e^{-xt} - 1) dt - \frac{1}{s}. $$
The difference is:
\begin{align*}
g(y) - g(0) &= \int_0^1 t^{-s-1} \left[ (e^{-xt}\cos(yt) - 1) - (e^{-xt} - 1) \right] dt \\
&= \int_0^1 t^{-s-1} e^{-xt} (\cos(yt) - 1) dt.
\end{align*}
We analyze the sign of the integrand, $K(t, y) = t^{-s-1} e^{-xt} (\cos(yt) - 1)$, for $t \in (0, 1]$.
Since $t > 0$, $t^{-s-1} > 0$. Since $x$ is real, $e^{-xt} > 0$. We know that $\cos(yt) \leq 1$, so $\cos(yt) - 1 \leq 0$.
Therefore, the integrand $K(t, y)$ is non-positive for all $t \in (0, 1]$.
This implies that the integral is non-positive:
$$ g(y) - g(0) \leq 0. $$
Hence, $g(y) \leq g(0)$ for all $y \in \mathbb{R}$. The real part of $f(z)$ is maximized when the imaginary part is zero.
\end{proof}


\subsection{Limiting Distribution for Truncated Sums of Inverse Exponentials}

Let $p \in (0, 2)$ and $R>0$. Let $\alpha = 2/p$. Since $p \in (0, 2)$, $\alpha \in (1, \infty)$.
Let $E$ be a standard exponential random variable.
We define a base random variable $Y = E^{-2/p} = E^{-\alpha}$.

We first calculate the distribution of $Y$. Let $F_Y(y)$ be the CDF of $Y$.
\focsarxiv{
\begin{align*}
F_Y(y) &= P(Y \le y) = P(E^{-\alpha} \le y) \\
&= P(E \ge y^{-1/\alpha}) = e^{-y^{-1/\alpha}}.
\end{align*}
}
{
\begin{align*}
F_Y(y) &= P(Y \le y) = P(E^{-\alpha} \le y) = P(E \ge y^{-1/\alpha}) = e^{-y^{-1/\alpha}}.
\end{align*}
}
The PDF of $Y$ is $f_Y(y) = F_Y'(y)$:
\[
f_Y(y) = \frac{1}{\alpha} y^{-1-1/\alpha} e^{-y^{-1/\alpha}}, \quad y>0.
\]
The tail is $1-F_Y(y) = 1 - e^{-y^{-1/\alpha}}$. 
So the distribution $Y$ has a heavy tail with index $1/\alpha = p/2$. 
By the theory of stable distributions, the normalized sum of i.i.d. copies of $Y$ converges to a $(p/2)$-stable law if the normalization factor is $k^{-\alpha} = k^{-2/p}$.

Let $X_1,\ldots, X_k$ be i.i.d. random variables distributed as $Y$ conditioned on $Y \le C_k$, where the truncation threshold is $C_k = R k^{\alpha}$.

Let the normalized sum be $S_k = k^{-\alpha} \sum_{i=1}^k X_i$. Let $\phi_k(t)$ be the CF of $S_k$.  Our main goal of this section is to calculate the limiting distribution of $S_k.$ 

We have
\[
\phi_k(t) = \mathbb{E}[e^{it S_k}] = \left( \mathbb{E}[e^{i (t k^{-\alpha}) X_1}] \right)^k.
\]
Let $\tau_k = t k^{-\alpha}$. We calculate the expectation $E_k = \mathbb{E}[e^{i \tau_k X_1}]$.
The PDF of $X_1$ is $f_X(x) = f_Y(x) / F_Y(C_k)$ for $0 < x \le C_k$.
\begin{align*}
E_k &= \frac{1}{F_Y(C_k)} \int_0^{C_k} e^{i \tau_k x} f_Y(x) dx \\
&= \frac{1}{F_Y(C_k)} \int_0^{C_k} (e^{i \tau_k x}-1+1) f_Y(x) dx \\
&= \frac{1}{F_Y(C_k)} \left( F_Y(C_k) + \int_0^{C_k} (e^{i \tau_k x}-1) f_Y(x) dx \right) \\
&= 1 + \frac{1}{F_Y(C_k)} \int_0^{C_k} (e^{i \tau_k x}-1) f_Y(x) dx.
\end{align*}

We analyze the limit of $\phi_k(t)$.
Let $J_k(t) = k (E_k-1)$.
\[
J_k(t) = \frac{k}{F_Y(C_k)} \int_0^{C_k} (e^{i \tau_k x}-1) f_Y(x) dx.
\]
As $k\to\infty$, $C_k \to \infty$, and $F_Y(C_k) = e^{-C_k^{-1/\alpha}} = e^{-(Rk^\alpha)^{-1/\alpha}} = e^{-R^{-1/\alpha}/k} \to 1$.
We analyze the limit of the remaining terms:
\[
C(t) = \lim_{k\to\infty} k \int_0^{C_k} (e^{i t k^{-\alpha} x}-1) f_Y(x) dx.
\]
We substitute $f_Y(x)$:
\[
C(t) = \lim_{k\to\infty} k \int_0^{Rk^\alpha} (e^{i t k^{-\alpha} x}-1) \frac{1}{\alpha} x^{-1-1/\alpha} e^{-x^{-1/\alpha}} dx.
\]
We perform a change of variables $z = k^{-\alpha} x$. $x = k^\alpha z$. $dx = k^\alpha dz$.
The upper bound becomes $Rk^\alpha / k^\alpha = R$.
\
\begin{align*}
J_k(t) &\approx k \int_0^{R} (e^{i t z}-1) \frac{1}{\alpha} (k^\alpha z)^{-1-1/\alpha} e^{-(k^\alpha z)^{-1/\alpha}} (k^\alpha dz) \\
&= k \int_0^{R} (e^{i t z}-1) \frac{1}{\alpha} k^{-\alpha-1} z^{-1-1/\alpha} e^{-k^{-1} z^{-1/\alpha}} k^\alpha dz \\
&= \int_0^{R} (e^{i t z}-1) \frac{1}{\alpha} z^{-1-1/\alpha} e^{-z^{-1/\alpha}/k} dz.
\end{align*}
As $k\to\infty$, the exponential term $e^{-z^{-1/\alpha}/k} \to 1$. We justify taking the limit inside the integral using the Dominated Convergence Theorem.
The integrand $g_k(z)$ is dominated by $G(z) = |e^{i t z}-1| \frac{1}{\alpha} z^{-1-1/\alpha}$.
Near $z=0$, $G(z) \sim \frac{|t|}{\alpha} z^{-1/\alpha}$. Since $1/\alpha = p/2 < 1$, $G(z)$ is integrable on $[0, R]$.

The limiting log-characteristic function $C(t) = \log \phi(t)$ is
\[
C(t) = \int_0^{R} (e^{i t z}-1) \frac{1}{\alpha} z^{-1-1/\alpha} dz.
\]

Next we evaluate the integral. Let $s = 1/\alpha = p/2$. Note $0 < s < 1$.
\[
C(t) = s \int_0^{R} (e^{i t z}-1) z^{-1-s} dz.
\]
We evaluate this integral using integration by parts. 
Let $u=e^{itz}-1$ and $dv=z^{-1-s}dz$, so that $du=it e^{itz}dz$ and $v=z^{-s}/(-s)$.
\begin{align*}
C(t) &= s \left[ (e^{itz}-1) \frac{z^{-s}}{-s} \right]_0^R - s \int_0^R it e^{itz} \frac{z^{-s}}{-s} dz \\
&= - \left[ (e^{itz}-1) z^{-s} \right]_0^R + it \int_0^R z^{-s} e^{itz} dz.
\end{align*}
At $R$: $-(e^{itR}-1) R^{-s}$.
At $z=0$: $|(e^{itz}-1) z^{-s}| \sim |itz| z^{-s} = |t| z^{1-s}$. Since $s<1$, this tends to 0 as $z\to 0$. 
Then we have:
\[
C(t) = R^{-s}(1-e^{itR}) + it \int_0^R z^{-s} e^{itz} dz.
\]

We relate the integral $K(t) = \int_0^R z^{-s} e^{itz} dz$ to the lower incomplete Gamma function $\gamma(a, Z) = \int_0^Z \tau^{a-1} e^{-\tau} d\tau$.
Let $a=1-s$. $0<a<1$. $K(t) = \int_0^R z^{a-1} e^{itz} dz$.
Using the substitution $\tau = -itz$. Let $Z=-itR$.
\[
K(t) = (-it)^{-a} \gamma(a, Z).
\]
Substituting this back into $C(t)$:
\begin{align*}
C(t) &= R^{-s}(1-e^{itR}) + it (-it)^{-a} \gamma(a, Z) \\
&= R^{-s}(1-e^{-Z}) - (-it)^{1-a} \gamma(a, Z) \\
&= R^{-s}(1-e^{-Z}) - (-it)^{s} \gamma(1-s, Z).
\end{align*}

To rewrite this, we use the recurrence relation for the lower incomplete gamma function:
$\gamma(a+1, Z) = a\gamma(a, Z) - Z^a e^{-Z}$.
Let $a=-s$.
$\gamma(1-s, Z) = -s \gamma(-s, Z) - Z^{-s} e^{-Z}$. We substitute this into the derived expression for $C(t)$. Note that $(-it)^s = (Z/R)^s$.
\focsarxiv{
\begin{align*}
&C(t)\\ &= R^{-s}(1-e^{-Z}) - (Z/R)^s \left[ -s \gamma(-s, Z) - Z^{-s} e^{-Z} \right] \\
&= R^{-s} - R^{-s}e^{-Z} + s (Z/R)^s \gamma(-s, Z) \\
&\hspace{5cm}+ (Z/R)^s Z^{-s} e^{-Z} \\
&= R^{-s} - R^{-s}e^{-Z} + s (Z/R)^s \gamma(-s, Z) + R^{-s} e^{-Z} \\
&= R^{-s} + s (Z/R)^s \gamma(-s, Z) \\
&= R^{-s} \left( 1 + s Z^s \gamma(-s, Z) \right).
\end{align*}
}
{
\begin{align*}
C(t) &= R^{-s}(1-e^{-Z}) - (Z/R)^s \left[ -s \gamma(-s, Z) - Z^{-s} e^{-Z} \right] \\
&= R^{-s} - R^{-s}e^{-Z} + s (Z/R)^s \gamma(-s, Z) + (Z/R)^s Z^{-s} e^{-Z} \\
&= R^{-s} - R^{-s}e^{-Z} + s (Z/R)^s \gamma(-s, Z) + R^{-s} e^{-Z} \\
&= R^{-s} + s (Z/R)^s \gamma(-s, Z) \\
&= R^{-s} \left( 1 + s Z^s \gamma(-s, Z) \right).
\end{align*}
}
\focsarxiv{
Substituting back $s=p/2$ and $Z=-itR$, $\phi(t)$ equals
\begin{align*}
\exp\left(R^{-p/2} \left(1+(p/2)(-i t R)^{p/2}\gamma(-p/2,-it R) \right)\right).
\end{align*}
}
{
Substituting back $s=p/2$ and $Z=-itR$, we have
\[
\phi(t) = \exp\left(R^{-p/2} \left(1+(p/2)(-i t R)^{p/2}\gamma(-p/2,-it R) \right)\right).
\]
}

\subsection{Inversion Formula Analysis}
Let $\phi(\xi)$ be the characteristic function derived above. 
Then by the Gil-Pelaez formula~\cite{wendel1961non}, our CDF is given by
\begin{align*}
F(t) &= \frac12 + \frac{1}{2\pi i}\int_0^{\infty} \frac{e^{it\xi}\phi(-\xi) - e^{-it\xi}\phi(\xi)}{\xi} d\xi\\
&:= \frac{1}{2} + \frac{1}{2\pi i}\int_{0}^{\infty} I_{t,R}(\xi) d\xi
\end{align*}
Our next step is to bound the integrand $I_{t,R}(\xi)$ on a strip about the real line.  This will allow us to apply known convergence bounds for the trapezoid method to show that this integral has an efficient numerical approximation.  For brevity we will drop the subscripts on $I$, but $R$ and $t$ are understood to be fixed parameters.

\begin{lemma}
\label{lem:bound_in_strip}
Let $I$ be the integrand in the Gil-Pelaez inversion formula above. Suppose that $w \leq \frac{1}{c_1(1 + t + R^{c_2})}.$  Then for $\xi$ in the strip $\{x : \Re(x)\geq 0, \abs{\Im(x)} \leq w\}$ we have the bound

\[
\abs{I(\xi)} \leq 
\exp(c_{p,1} - c_{p,2} R^{-p/2}\abs{\xi}^{p/2})(1 + \poly(R) + t).
\]
\end{lemma}

\begin{proof}

We would like to bound this integrand on the strip $\{z: \Re(z)\geq 0, \abs{\Im(z)} \leq w\}$ around the positive real line in the complex plane. 
We will always assume that $w\leq 1.$

Substituting in $\phi$, we rewrite the numerator of the integrand as 
\focsarxiv{
\begin{align*}
&\exp(it\xi + R^{-p/2} (1+(p/2)(i \xi R)^{{p/2}}\gamma(-p/2,i\xi R) )) \\
&-\exp(-it\xi + R^{-p/2}\\
&\hspace{0.2in}\cdot(1+(p/2)(-i \xi R)^{{p/2}}\gamma(-p/2,-i\xi R) ) )
\\
&:= \exp(f(\xi)) - \exp(f(-\xi)),
\end{align*}
}
{
\begin{align*}
&\exp(it\xi + R^{-p/2} (1+(p/2)(i \xi R)^{{p/2}}\gamma(-p/2,i\xi R) )) \\
&-\exp(-it\xi + R^{-p/2}(1+(p/2)(-i \xi R)^{{p/2}}\gamma(-p/2,-i\xi R) ) )
\\
&:= \exp(f(\xi)) - \exp(f(-\xi)),
\end{align*}
}
where
\[
f(\xi) = it\xi + R^{-p/2} \left(1+\frac{p}{2}(i \xi R)^{{p/2}}\gamma(-p/2,i\xi R) \right).
\]
Then our integrand can be rewritten as
\[
\exp(f(-\xi))\cdot\left(\frac{\exp(f(\xi) - f(-\xi)) - 1}{\xi}\right).
\]
We will bound these two terms.

\paragraph{First term.}

We would like to bound $\Re(f(\xi))$ for $\xi$ in the strip.  For the first term of $f$, observe that $\Re(it\xi) = O(1)$ in the strip. The second term of $f$ is maximized for purely imaginary $\xi$ by Proposition~\ref{prop:max_imaginary_part_zero}.  So it suffices to bound this quantity for $\xi$ in the ball of radius $w$ in the right half plane. 

To do this, we use the power series expansion of $\gamma(-p/2, z)$ near $0$ which gives 
\[
1 + \frac{p}{2}z^{p/2}\gamma(-p/2, z) = O(z),
\]
as $\abs{z} \rightarrow 0.$  This means that $f(\xi) \leq O(1 + w R^{1 - p/2}) \leq O(1)$ on a $w$-neighborhood of $0.$  By the above, this bound hold uniformly on the strip.

To analyze the behavior away from zero, we use the asymptotic expansion of the lower incomplete gamma function to obtain
\focsarxiv{
\begin{align*}
1 &+ \frac{p}{2} z^{p/2}\gamma(-p/2, z)\\
&= 1 + \frac{p}{2}\Gamma(-p/2)z^{p/2} + O(z^{-1}e^{-z})\\
&= 1 - c_p z^{p/2} + O(z^{-1}e^{-z})
\end{align*}
}
{
\begin{align*}
1 + \frac{p}{2} z^{p/2}\gamma(-p/2, z) &= 1 + \frac{p}{2}\Gamma(-p/2)z^{p/2} + O(z^{-1}e^{-z})\\
&= 1 - c_p z^{p/2} + O(z^{-1}e^{-z})
\end{align*}
}
as $\abs{z} \rightarrow \infty$ in the sector $\abs{\arg(z)} \leq 3\pi/2,$ and where $c_p$ is a positive constant.




Together these bounds imply that 
there are positive constants $c_p$ and $c_p'$
such that for $\abs{\Im(\xi)} \leq 1/R$, we have
\focsarxiv{
\begin{align*}
\Re(1 + \frac{p}{2}(i\xi R)^{p/2} &\gamma(-p/2, i\xi R)) \\
&\leq \min(c_p, 1 - c_p' \abs{i \xi R}^{p/2}).
\end{align*}
}
{
\[
\Re(1 + \frac{p}{2}(i\xi R)^{p/2} \gamma(-p/2, i\xi R)) \leq \min(c_p, 1 - c_p' \abs{i \xi R}^{p/2}).
\]
}

Thus for $\abs{\Im(\xi)} \leq \frac{1}{t}$ we have 
\[
\abs{\exp(f(-\xi))} \leq \exp(c_p R^{-p/2} - c_p' \xi^{p/2})
\]

\paragraph{Second term.} For the second term, set $h(z) = z^{p/2} \gamma(-p/2, z)$ and note that
\[
f(\xi) - f(-\xi)
= 
2it\xi  + \frac{p}{2} R^{-p/2}(h(i\xi R) - h(-i\xi R)).
\]

Since $\gamma(-p/2, z)$ is real-valued on the real line, we have
\begin{align*}
\Re(h(z) - h(-z))
&= \Re(h(z) - \overline{h(z)} + h(\overline{z}) - h(-z))\\
&= \Re(h(\overline{z}) - h(-z)).
\end{align*}
To bound this, we have
\begin{align*}
\abs{\Re(h(\overline{z}) - h(-z))} &= \abs{\Re(h(\overline{z}) - h(-z)} \\
&\leq
\abs{h(\overline{z}) - h(-z)}
\leq \abs{\int_{-z}^{\overline{z}} h'(t) dt},
\end{align*}
where the integral is along a path joining $-z$ and $\overline{z}.$ Note that $\abs{(-z) - \overline{z}} = 2\abs{\Re(z)}$.  On the horizontal segment joining $-z$ and $\overline{z}$, we have the bound $\abs{h'(t)} \leq C_p$ by \Cref{prop:h_bounded_deriv}, and so the integral is bounded by $2 C_p \abs{\Re(z)}$.

To bound $h(\overline{z}) - h(-z)$  near $0$, recall that we can write the lower incomplete gamma function as
\[
\gamma(s,x) = z^s \Gamma(s)\gamma_*(s,z),
\]
where $\gamma_*$ is entire in $z.$  
So $h(z) = \Gamma(-p/2)\gamma_*(-p/2,z).$  
Note that $(h(z) - h(-z))/z$ is entire in $z$.  In particular, it is continuous in $z$ and therefore bounded by a constant $C_p$ inside the unit disk.

So, we have that $\abs{h(z) - h(-z)}\leq C_p\abs{z}$ in the unit disk, and elsewhere $\Re(h(z) - h(-z)) \leq c\abs{\Re(z)}$. 
Then for $\abs{\xi} \leq \frac{1}{R^c},$
\[\abs{h(-i\xi R) - h(i\xi R)} \leq C_p R \abs{\xi},\]
and elsewhere,
\[
\Re(h(-i\xi R) - h(i\xi R))
\leq c \Re(-i\xi R)
= c R \abs{\Im(\xi)}.
\]
For $\abs{\Im(\xi)} \leq \min(\frac{1}{t}, \frac{1}{\poly(R)}) $ and $\abs{\xi} \geq 1/R,$ we have
\focsarxiv{
\begin{align*}
\abs{\frac{\exp(f(\xi) - f(-\xi)) - 1}{\xi}}\\
&\hspace{-4cm}= \abs{\frac{\exp(2it\xi) \exp(\frac{p}{2} R^{-p/2}(h(i\xi R) - h(-i\xi R))) - 1}{\xi}}\\
&\hspace{-4cm}\leq R (\exp(2t\Im(\xi) + C_p R^{-p/2} R \abs{\Im(\xi)}) + 1)\\
&\hspace{-4cm}\leq C_p R,
\end{align*}
}
{
\begin{align*}
\abs{\frac{\exp(f(\xi) - f(-\xi)) - 1}{\xi}}
&= \abs{\frac{\exp(2it\xi) \exp(\frac{p}{2} R^{-p/2}(h(i\xi R) - h(-i\xi R))) - 1}{\xi}}\\
&\leq R (\exp(2t\Im(\xi) + C_p R^{-p/2} R^c \Im(\xi)) + 1)\\
&\leq C_p R^c,
\end{align*}
}
for a constant $C_p>0$ depending only on $p$.

For $\abs{\xi} \leq \frac{1}{\poly(R)}$ and $\abs{\Im(\xi)} \leq \frac{1}{t}$, we have
\focsarxiv{
\begin{align*}
\abs{\frac{\exp(f(\xi) - f(-\xi)) - 1}{\xi}}\\
&\hspace{-4cm}= \abs{\frac{\exp(2it\xi + \frac{p}{2} R^{-p/2}(h(i\xi R) - h(-i\xi R))) - 1}{\xi}}\\
&\hspace{-4cm}= \Bigg\vert\exp(2it\xi)\cdot\left(\frac{\exp(\frac{p}{2} R^{-p/2}(h(i\xi R) - h(-i\xi R))) - 1}{\xi}\right) \\
&+ \frac{\exp(2it\xi) - 1}{\xi}\Bigg\vert\\
&\hspace{-4cm}\leq \abs{\exp(2it\xi)} \abs{\frac{\exp(\frac{p}{2} R^{-p/2}(h(i\xi R) - h(-i\xi R))) - 1}{\xi}}\\ 
&+ \abs{\frac{\exp(2it\xi) - 1}{\xi}}\\
&\hspace{-4cm}\leq C (1 + R^{1-p/2} + t)\\
&\hspace{-4cm}\leq C(1 + R + t).
\end{align*}
}
{
\begin{align*}
\abs{\frac{\exp(f(\xi) - f(-\xi)) - 1}{\xi}}
&= \abs{\frac{\exp(2it\xi + \frac{p}{2} R^{-p/2}(h(i\xi R) - h(-i\xi R))) - 1}{\xi}}\\
&= \abs{\exp(2it\xi)\cdot\left(\frac{\exp(\frac{p}{2} R^{-p/2}(h(i\xi R) - h(-i\xi R))) - 1}{\xi}\right) + \frac{\exp(2it\xi) - 1}{\xi}}\\
&\leq \abs{\exp(2it\xi)} \abs{\frac{\exp(\frac{p}{2} R^{-p/2}(h(i\xi R) - h(-i\xi R))) - 1}{\xi}} + \abs{\frac{\exp(2it\xi) - 1}{\xi}}\\
&\leq C (1 + R^{1-p/2} + t)\\
&\leq C(1 + R + t).
\end{align*}
}
given the stated bounds on $\xi.$

Combining the bounds above gives the stated result.
\end{proof}

\subsection{Evaluating the Integral Numerically}

This allows us to show that a quadrature rule for the integral converges quickly.  Recall that our CDF was given by
\[
F(t) = \frac12 + \frac{1}{2\pi i}\int_0^{\infty} \frac{e^{it\xi}\phi(-t) - e^{-it\xi}\phi(t)}{\xi} d\xi.
\]
To evaluate this integral, we apply a variant of the trapezoid rule. 
We have previously bounded the integrand off the real line, which is what is needed to guarantee that the trapezoid rule to converge at an exponential rate. 

First recall the simplest numerical integral scheme, 
namely approximating the integral $I = \int_{-\infty}^{\infty}f(x) dx$ by
\[
I_h = h\sum_{i=-\infty}^{\infty} f(x_i)
\]
where $x_{i+1} = h + x_i$, so that the $x_i$'s form a fixed mesh with separation $h.$  We will use the following convergence bound.

\begin{theorem}
\label{thm:trapezoid_convergence}
(\cite{trefethen2014exponentially} Theorem 5.1) Suppose that $f$ is analytic in the strip $\Im(z) \leq a$ for some $a>0.$  Suppose that $f(x)\rightarrow 0$ uniformly as $\abs{x}\rightarrow\infty$ in the strip, and that for some $M$,
\[
\int_{-\infty}^{\infty}\abs{f(x + ib)} dx \leq M,
\]
for all $b\in (-a,a).$  Then 
\[
\abs{I_h - I} \leq \frac{2M}{e^{2\pi a/h} - 1}.
\]
\end{theorem}

Now we use this along with our previous bounds for our lemma.

\begin{lemma}
For fixed $R\geq \frac{1}{\log n}$ $t$, there is an algorithm that evaluates $F(t)$ to within $\eps$ additive error using $\poly(t, R, \log\frac{1}{\eps})$ evaluations of the integrand $I(t)$.
\end{lemma}

\begin{proof}
We first apply a transformation to our integral to move the branch point at $0$ to infinity.  We will use a version of the $\tanh$ quadrature rule here, although there are other possibilities.

We will truncate our integral to a finite range $[0,L]$, so that the integral over $[0,\infty)$ is approximated to within $\eps$.  Using the decay on $I(t)$ given by \Cref{lem:bound_in_strip}, we can take $L = \O{\poly(R,t, \log\frac{1}{\eps})}$.

We use the transformation $t \mapsto \frac{1}{2}L\cdot(1 + \tanh(t))$ to rewrite the integral:
\begin{align*}
\int_{0}^L I(t) dt
&= \frac{L}{2}\int_{-\infty}^{\infty} I(L\cdot(1 + \tanh(t)) \sech^2(t) dt\\
&:= \frac{L}{2}\int_{-\infty}^{\infty} \tilde{I}(t) dt.
\end{align*}
From \Cref{lem:bound_in_strip} we have that $I(t) \leq C_{t,R}$ in the strip of width $w$ about the real axis.

For $w < \pi/2$, the function $\tanh$ maps the strip  $S_w = \{z:\abs{\Im{z}} \leq w\}$ into the strip $S_{\tan w}$ and in fact it maps to a lens contained in this strip.  
Therefore, if $z\in S_w$, then $L(1+\tanh(z))$ is contained in $S_{L\tan(w)},$ which is contained in $S_{2Lw}$ for $w\leq 1.$  
We will apply \Cref{thm:trapezoid_convergence} on the strip of width $a = \frac{1}{c(1+\abs{t} + \abs{R})}\frac{1}{2L}.$  
For $z$ in the strip $S_a$ we then have $\abs{\tilde{I}(z)} \leq C_{t,r} \abs{\sech^2(z)}.$  
Since $\sech(z)$ decays exponentially on the strip as $\abs{z}\rightarrow\infty$, we can take $M = \O{C_{t,r}}$ in \Cref{thm:trapezoid_convergence}. 
This means that we have
\[\abs{I_h - I} \leq \O{\frac{C_{t,R}}{e^{2\pi a / h} - 1}},\]
which is upper bounded by $\eps$ for 
\begin{align*}\frac{1}{h} &= \O{L \abs{t} + \abs{\log\frac{1}{R}}}\cdot\log C_{t,r} \log\frac{1}{\eps }\\
&= \O{(\abs{t} \polylog(t,r,\eps)}.
\end{align*}
Since this integrand decays exponentially on both sides, we may approximate it to within $\eps$ by restricting to a domain of width $\O{\log C_{t,R}/\eps}$.  Thus we only require $\poly(t, R, \log\frac{1}{\eps})$ evaluations of the integrand.

\end{proof}

\subsection{Derandomizing the \texorpdfstring{$L_p$}{Lp} Sampler}
The previous analysis assumes that independent exponential random variables can be efficiently generated and stored. 
To achieve a true streaming algorithm, we now turn our attention to derandomizing these elements of the linear sketch. 
We first recall the PRG of \cite{GopalanKM18} that fools a certain class of Fourier transforms. 
\begin{definition}
An \emph{$(m,n)$-Fourier shape} is a function $f:[m]^n \to \mathbb{C}$ of the form
\[f(x_1,\ldots,x_n) = \prod_{j=1}^n f_j(x_j),\]
where each $f_j:[m]\to\mathbb{C}$.
\end{definition}
\noindent
We have the following properties of the Fourier-shape fooling PRG, c.f., Theorem 1.1 of \cite{GopalanKM18}.
\begin{theorem}
\cite{GopalanKM18}
\label{thm:gkm:prg}
Let $f:[m]^n\to\mathbb{C}$ be an $(m,n)$-Fourier shape, where $f(x_1,\ldots,x_n)=\prod_{j=1}^n f_j(x_j)$ and each $f_j:[m]\to\mathbb{C}$. 
Then for any $\eps>0$, there exists a seed length $\ell=\O{\log\frac{mn}{\eps}\cdot\left(\log\log\frac{mn}{\eps}\right)^2}$ and a deterministic function $G:\{0,1\}^\ell\to[m]^n$ such that 
\[\left\lvert\EEx{x\sim[m]^n}{f(x)}-\EEx{y\sim\{0,1\}^\ell}{f(G(y))}\right\rvert\le\eps.\]
Equivalently, all $(m,n)$-Fourier shapes are $\eps$-fooled by $G$.
\end{theorem}

Moreover, the PRG of \cite{GopalanKM18} fools the class of half-space queries, defined as follows. 
For each $i\in[\lambda]$, let a half-space query $H_i:\mathbb{R}^n\to\{0,1\}$ on an input $X=(x_1,\ldots,x_n)$ be defined by the indicator function $\mathbf{1}\left[a_1^{(i)}x_1+\ldots+a_n^{(i)}x_n\right]>\theta_i$, where $a^{(i)}\in\mathbb{Z}^n$ and $\theta_i\in\mathbb{Z}$ for all $i\in[n]$. 
\begin{definition}[$\lambda$-half-space tester]
Given an input $X=(x_1,\ldots,x_n)$, a $\lambda$-half-space tester is a function $\sigma(H_1(X),\ldots,H_\lambda(X))\in\{0,1\}$ where $\sigma:\{0,1\}^\lambda\to\{0,1\}$ and $H_1,\ldots,H_\lambda(X)$ are half-space queries. 
We say the $\lambda$-half-space tester is $M$-bounded if each input coordinate $x_i$ is drawn from a distribution on the integers with magnitude at most $M$, and if all coefficients $a_j^{(i)}$ and thresholds $\theta_i$ have magnitude at most $M$. 
\end{definition}
We have the following properties of the half-space fooling PRG, c.f., Lemma 7 and Proposition 8 of \cite{JayaramW18}.
\begin{theorem}
\cite{GopalanKM18,JayaramW18}
\label{thm:gkm:prg:half:space}
Let $\calD$ be a distribution on the integers $\{-M,\ldots,M\}$ with magnitude at most $M$ that can be sampled using at most $\O{\log M}$ random bits. 
Let $x_1,\ldots,x_n\sim\calD$ and let $X=(x_1,\ldots,x_n)$. 
Then for any $\eps>0$ and $c\ge 1$, there exists $\ell=\O{\lambda\left(\log\frac{nM}{\eps}\right)\left(\log\log\frac{nM}{\eps}\right)^2}$ and a deterministic function $F:\{0,1\}^\ell\to\{-M,\ldots,M\}^n$ such that for every set of $\lambda$ half-space tester $\sigma_H$, we have
\[\left\lvert\EEx{X\sim\calD^n}{\sigma_H(X)=1}-\EEx{y\sim\{0,1\}^\ell}{\sigma_H(F(y))=1}\right\rvert\le\eps.\]
Moreover, if $F(y)\in\{-M,\ldots,M\}^n$ then each coordinate of $F(y)$ can be computed in $\O{\ell}$ space and $\polylog(nM)$ time. 
\end{theorem}
For convenience of discussion, we shall refer to both the Fourier-shape fooling PRG and the half-space fooling PRG by \cite{GopalanKM18} as the GKM PRG. 
We now describe the construction of the GKM PRG, though for the purposes of derandomizing our algorithm, the guarantees of \Cref{thm:gkm:prg} and \Cref{thm:gkm:prg:half:space} suffice. 
The GKM PRG is composed of three main components, each addressing different cases for the variance of the Fourier coefficients in the tester. 
The first component handles the high variance cases where there can be many nonzero Fourier coefficients; it uses a pairwise hash function for subsampling and applies a PRG by Nisan and Zuckerman~\cite{nisan1996randomness} to derandomize the selection process, allowing the generator to produce pseudorandom output in linear time and space relative to the seed. 
The second component handles the high variance case where there can be a small number of large Fourier coefficients; it reduces the problem size to smaller coefficients by iteratively composing smaller instantiations of Nisan's PRG~\cite{Nisan92}.  
Finally, the third component focuses on low variance cases and reduces the number of coordinates while expanding the range of values by using $k$-wise independent hash functions to handle the reduced coordinate space and a recursive structure to handle varying levels of complexity.

Additionally, we require the following derandomization property from \cite{JayaramW18}:
\begin{theorem}
\label{thm:jw:prg}
\cite{JayaramW18}
Let $\calS$ be any streaming algorithm that, given a stream vector $f\in\{-M,\ldots,M\}^n$ and a fixed matrix $X\in\mathbb{R}^{t\times nk}$ with entries from $\{-M,\ldots,M\}$ for some $M=\poly(n)$, produces an output determined solely by the sketches $A f$ and $X\cdot\vec{A}$. 
Suppose the random matrix $A\in\mathbb{R}^{k\times n}$ has i.i.d.\ entries, each samplable using $\O{\log n}$ bits. 
Let $\sigma:\mathbb{R}^k\times\mathbb{R}^t\to\{0,1\}$ be a tester that verifies the success of $\calS$, so that $\sigma(Af,X\vec(A))=1$ whenever $\calS$ succeeds. 
For any fixed constant $c\ge1$, there exists an implementation of $\calS$ using a random matrix $A'$ that can be generated from a seed of length $\O{(k+t)\log n(\log\log n)^2}$ satisfying
\focsarxiv{
\begin{align*}
&\left\lvert\PPr{\sigma(Af,X\vec(A))=1}-\PPr{\sigma(A'f,X\vec(A'))=1}\right\rvert\\
&<n^{-c(k+t)}.
\end{align*}
}
{
\begin{align*}
\left\lvert\PPr{\sigma(Af,X\vec(A))=1}-\PPr{\sigma(A'f,X\vec(A'))=1}\right\rvert<n^{-c(k+t)}.
\end{align*}
}
Moreover, each entry of $A'$ can be computed in $\polylog(n)$ time using working space proportional to the seed length.
\end{theorem}
First, we recall that our algorithm uses two sources of randomness: the Gaussian random variables in the dense $\CountSketch$ and the exponential random variables to form the scalings $\frac{x_i}{e_{i,j}^{2/p}}$. 
We shall use separate instantiations of the GKM PRG to derandomize the generation of the Gaussian random variables and to derandomize the generation of the exponential random variables. 
To that end, let $R_g$ denote the randomness used to generate the Gaussian random variables and let $R_e$ denote the randomness used to generate the exponential random variables, so that both $R_g$ and $R_e$ require $\poly(n)$ bits. 
Now for any fixed $i\in[n]$, our algorithm can be viewed as a tester $\calA_i(R_g,R_e)\in\{0,1\}$ for whether the $L_p$ sampler will output the index $i$ as the sampled coordinate. 
Our goal is to show that there exist two instantiations $F_1(\cdot)$ and $F_2(\cdot)$ of the GKM PRG so that
\focsarxiv{
\begin{align*}
&\left\lvert\PPPr{R_g,R_e}{\calA_i(R_g,R_e)=1}- \PPPr{x,y}{\calA_i(F_1(y_1),F_2(y_2))=1}\right\rvert\\
&\le \frac{1}{n^C}.
\end{align*}
}
{
\[\left\lvert\PPPr{R_g,R_e}{\calA_i(R_g,R_e)=1}-\PPPr{x,y}{\calA_i(F_1(y_1),F_2(y_2))=1}\right\rvert\le\frac{1}{n^C},\]
}
where $C$ is an arbitrarily large constant and $y_1$ and $y_2$ are seeds of length $\O{\log^2 n}$ that are sufficiently large for the choice of $C$. 

\begin{theorem}
\algref{alg:perfect:lp:sample} can be derandomized to use $\O{\log^2 n(\log\log n)^2}$ bits of space for $p<2$. 
Moreover, the derandomized algorithm uses $\polylog(n)$ update time per arriving element in expectation. 
\end{theorem}
\begin{proof}
Our goal is to derandomize the generation of $\poly(n)$ Gaussian random variables in the dense $\CountSketch$ data structure and $\poly(n)$ exponential random variables, including the head and the tail of the order statistics by the simulation oracle. 
For any fixed randomness $R_e$ used to generate the exponential random variables (including the tail statistics which are then simulated by generating separate Gaussian random variables) and any fixed $i\in[n]$, let $\calA_{i,R_e}(R_g)$ be the tester that determines whether our $L_p$ sampler selects index $i$. 
Here, $R_e$ is hard-wired into the input while $R_g$ is provided is provided as a stream of randomness. 
We show that $\calA_{i,R_e}(R_g)$ can be fooled by using a seed of length $\O{\log^2 n(\log\log n)^2}$ for $p<2$. 

Because $R_e$ is hard-wired into the input, then we can write the linear sketch as $GZx=Gz$, where $Z\in\mathbb{R}^{n^{c+1}\times n}$ consists of the scalings due to the exponential random variables and $z=Zx$, which we can compute since $R_e$ is hard-coded into the input. 
Note that we cannot immediately argue about linear threshold functions because the estimator considers $\langle G_i,Gz\rangle$, which is a quadratic threshold function. 
Instead, we consider fooling a sufficient number of half-space queries as follows. 
Observe that $G\in\mathbb{R}^{k\times n^{c+1}}$ for $k=\O{\log n}$ and has i.i.d. entries. 
Due to the probability density function of normal random variables, we first condition on the high-probability event that all sampled Gaussians have magnitude at most $n^{\O{c}}$. 
Due to the truncation of the sampled Gaussians, it follows that each random Gaussian can be sampled using $\O{\log n}$ bits of randomness and the entire algorithm only requires $\poly(n)$ random bits to be generated for the Gaussians. 
It suffices to test whether $Gz=u$ across only the vectors $u\in\mathbb{R}^{\O{\log n}}$ with integer entries bounded by $\O{\log n}$ bit complexity, since $G$ has $\O{\log n}$ rows. 
This corresponds to a set of $\O{\log n}$ half-space testers and more generally, we argue that we can fool $\O{n^{c+1}}$ sets of $\O{\log n}$ half-space testers. 
By applying \Cref{thm:gkm:prg:half:space}, we can use the GKM PRG to construct a deterministic algorithm $F_1$ that fools any efficient tester that outputs a function of $Gz$ and the vectorization of $G$, using a seed $y_1$ of length $\O{\log^2 n(\log\log n)^2}$ for $p<2$. 
Now, since we can fool $\PPr{Gz=u}$, we can also fool any tester which takes as input $u=Gz$ and outputs whether or not on input $u$, the algorithm would output $i\in[n^{c+1}]$, as in \Cref{thm:jw:prg}. 
Specifically, fix any $i\in[n^{c+1}]$, so that the output of dense $\CountSketch$ requires exactly both the column $G_i$ and the value of $Gz$. 
The column $G_i$ has $\O{\log n}$ entries, each of which can be expressed in $\O{\log n}$ bits, and thus can be written in terms of $X\cdot\vec{G}$. 
Crucially, the seed length required to fool such a tester depends on the sketch size rather than the length of the input vector. 
Therefore, the tester $\calA_{i,R_e}(R_g)$ can be fooled using a seed of the same length:
\focsarxiv{
\begin{align*}
\PPr{\calA_{i}(R_g,R_e)=1}\\
&\hspace{-9em}= \sum_{R_e}\calA_{i,R_e}(R_g)\cdot\PPr{R_e}\\
&\hspace{-9em}= \sum_{R_e}\Bigl(\PPr{\calA_{i,R_e}(F_1(y_1))=1} \pm \frac{1}{n^{\O{\log n}}}\Bigr)\cdot \PPr{R_e}\\
&\hspace{-9em}= \sum_{R_e}\PPr{\calA_{i,R_e}(F_1(y_1))=1}\cdot \PPr{R_e} \pm \frac{1}{n^{\O{\log n}}}\\
&\hspace{-9em}= \PPr{\calA_i(F_1(y_1),R_e)=1} \pm \frac{1}{n^{\O{\log n}}}.
\end{align*}
}
{
\begin{align*}
\PPr{\calA_{i}(R_g,R_e)=1}&=\sum_{R_e}\calA_{i,R_e}(R_g)\cdot\PPr{R_e}\\
&=\sum_{R_e}\left(\PPr{\calA_{i,R_e}(F_1(y_1))=1}\pm\frac{1}{n^{\O{\log n}}}\right)\cdot\PPr{R_e}\\
&=\sum_{R_e}\PPr{\calA_{i,R_e}(F_1(y_1))=1}\cdot\PPr{R_e}\pm\frac{1}{n^{\O{\log n}}}\\
&=\PPr{\calA_i(F_1(y_1),R_e)=1}\pm\frac{1}{n^{\O{\log n}}}.
\end{align*}
}
Finally to account for all indices $i\in[n^{c+1}]$ and columns $G_i$, we can union bound over all such $\O{n^{c+1}}$ sets of testers and show that the total variation distance remains $\frac{1}{n^{\O{\log n}}}$. 

Now, consider any fixed seed $F_1(y_1)$ and corresponding tester $\calB_{i,F_1(y_1)}(R_e)$. 
Note that this determines the sketch matrix $G' = F_1(y_1)$ and so it suffices to consider the residual tester $\calB_i(R_e) := A_i(G', R_e)$. 
We further decompose the randomness $R_e$ for generating the exponential random variables into the randomness $R_h$ for generating the head $\{\calH_i\}$ of the order statistics and the randomness $R_t$ for generating the tail $\{N_{i,j}\}$ of the order statistics. 
Because we truncate each of the coordinates to $\frac{1}{\poly(n)}$ granularity, it suffices to only focus on the truncated exponential random variables on a set with support $\poly(n)$. 
Indeed, due to the probability density function of exponential random variables, we can condition on the high-probability event that all sampled exponential random variables have magnitude at most $n^{\O{c}}$ and at least $\frac{1}{n^{\O{c}}}$, which only distorts the output distribution of the $L_p$ sampler by an additive $\frac{1}{n^{\O{c}}}$. 
Moreover, the largest $\tau$ values in the set $\calH_i$ correspond to the smallest $\tau$ values, which can then be generated by using $\tau$ independent exponential random variables due to the order statistics in \Cref{lem:hidden:exps}. 
Hence, the head statistics correspond to $\O{\log n}$ values that can each be generated using $\O{\log n}$ bits. 
Moreover, conditioned on the randomness $R_g$ and $R_t$, the output of the dense $\CountSketch$ and the statistical test can be determined from running $\O{\log n}$ half-space tests, using an argument similar to the proof of Theorem 7 in \cite{JayaramW18} and thus we can apply the GKM PRG to derandomize $R_e$ using a seed of length $\O{\log^2 n(\log\log n)^2}$. 

It remains to derandomize $R_t$. 
We apply another hybrid argument to hard-code $R_h$ and thus condition on the largest $\tau$ values. 
We use $\O{\log n}$ bits of randomness to simulate the sum of the remaining exponential scalings by generating each $N_{i,j}$. 
In particular, given the value of the largest $\tau$ order statistics, the simulation oracle will perform a fixed number of deterministic calculations to sample the tail statistics from the CDF through a binary search. 
The random sampling step is implemented by drawing a random variable $U$ uniformly at random from $(0,1)$ and locating the point whose CDF value equals or exceeds $U$, thereby determining the appropriate tail statistic. 
In summary, the simulation oracle determines the remaining tail statistic by performing an inverse transform, which, given the fixed CDF and the largest $\tau$ order statistics, can be deterministically implemented using a fixed number of binary search steps corresponding to half-space queries on the CDF's domain. 
Afterwards, the residual tester $\calB_i(R_t)$ computes the sketch vector $Y(R_t) = G'\cdot z(R_t)$, where $z(R_t)$ is the vector of simulated variables with dimension $n' = \O{n\log n}$. 
The tester subsequently applies the statistical test $T_{\mathrm{stat}}$ in \Cref{fig:stat:test} and checks whether the output is some value $i\in[n]$ or FAIL. 

Let $P_Y$ denote the distribution of $Y(R_t)$ when $R_t$ is drawn uniformly at random. 
We seek a pseudorandom generator $F_2$ that produces $R'_t = F_2(y_2)$ for a seed $y_2$, such that the induced distribution $P'_Y$ of $Y(R'_t)$ is close to $P_Y$, i.e.,  
\focsarxiv{
\begin{align*}
\left|\EEx{R_t}{\calB_i(R_t)} - \EEx{y_2}{\calB_i(F_2(y_2))}\right| &\le 2\, d_{\mathrm{TV}}(P_Y, P'_Y)\\
&\le\frac{1}{\poly(n)}
\end{align*}
}
{
\[\left|\EEx{R_t}{\calB_i(R_t)} - \EEx{y_2}{\calB_i(F_2(y_2))}\right| \le 2\, d_{\mathrm{TV}}(P_Y, P'_Y)\le\frac{1}{\poly(n)}.\]
}
Since the vector $Y$ has $R = \O{\log n}$ coordinates, each discretized to $L=\O{\log n}$ bits, then the support size satisfies
\[N_{\mathrm{supp}} \le (2^L)^R = (n^{\O{C'}})^{\O{\log n}} = n^{\O{\log n}}.\]
For distributions $P,Q$ on a finite domain $\Omega$, the standard relationship between total variation and Fourier distance $d_{\mathrm{FT}}$ is
\[d_{\mathrm{TV}}(P,Q) \le \frac{1}{2}\sqrt{|\Omega|}\,\|P-Q\|_2.\]
By Plancherel’s theorem, the $\ell_2$-norm difference satisfies $\|P-Q\|_2^2 \le (d_{\mathrm{FT}}(P,Q))^2$, so that
\[d_{\mathrm{TV}}(P,Q) \le \frac{1}{2}\sqrt{N_{\mathrm{supp}}}\, d_{\mathrm{FT}}(P,Q).\]
To guarantee $d_{\mathrm{TV}}(P_Y, P'_Y) \le\frac{1}{\poly(n)}$, it suffices that
\[
\eps_{\mathrm{FT}} := \frac{n^{-\O{1}}}{\sqrt{N_{\mathrm{supp}}}} = \frac{n^{-\O{1}}}{n^{\O{\log n}}} = n^{-\O{\log n}}.
\]
We now analyze the characteristic function of $Y(R_t)$. 
For $\alpha \in \mathbb{R}^R$,
\begin{align*}
\Phi_Y(\alpha) &= \EEx{R_t}{\exp(2\pi i\langle \alpha, Y(R_t)\rangle)} \\
&= \EEx{R_t}{\exp(2\pi i\langle (G')^\top\alpha, z(R_t)\rangle)}.
\end{align*}
Let $C_\alpha = (G')^\top\alpha$. 
Since $R_t = (r_1,\ldots,r_{n'})$ contains independent randomness for each coordinate of $z$, we obtain
\[\Phi_Y(\alpha) = \EEx{R_t}{\prod_{j=1}^{n'} \exp(2\pi i\, C_{\alpha,j}\, z_j(r_j))}.\]
The inner function is a Fourier shape of the form $F_\alpha(R_t) = \prod_{j=1}^{n'} f_j(r_j)$, where each $f_j(r_j)$ maps into the complex unit disk.
We apply the GKM PRG to fool this Fourier shape using parameters
\focsarxiv{
\[N = n' = \O{n\log n}, \quad M = 2^L = \polylog(n)\]
\[\eps = \eps_{\mathrm{FT}} = n^{-\O{\log n}}.\]
}
{
\[N = n' = \O{n\log n}, \quad M = 2^L = \polylog(n), \quad \eps = \eps_{\mathrm{FT}} = n^{-\O{\log n}}.\]
}
By \Cref{thm:gkm:prg}, the corresponding seed length is $r_2 = \O{\log\frac{NM}{\eps}\left(\log\log\frac{NM}{\eps}\right)^2}$. 
Since
\focsarxiv{
\begin{align*}
\log(NM/\eps_{\mathrm{FT}}) &= \O{\log n} + \O{\log n} + \O{\log^2 n} \\
&= \O{\log^2 n},
\end{align*}
}
{
\[\log(NM/\eps_{\mathrm{FT}}) = \O{\log n} + \O{\log n} + \O{\log^2 n} = \O{\log^2 n},\]
}
we have $r_2 = \O{\log^2 n(\log\log n)^2}$. 
This pseudorandom generator $F_2$ ensures $d_{\mathrm{FT}}(P_Y, P'_Y) \le \eps_{\mathrm{FT}}$, and therefore $d_{\mathrm{TV}}(P_Y, P'_Y) \le\frac{1}{\poly(n)}$. 
Consequently,
\focsarxiv{
\begin{align*}
&\left|\EEx{y_1,R_t}{A_i(F_1(y_1),R_t)} - \EEx{y_1,y_2}{A_i(F_1(y_1),F_2(y_2))}\right| \\
&\le\frac{1}{\poly(n)},
\end{align*}
}
{
\[\left|\EEx{y_1,R_t}{A_i(F_1(y_1),R_t)} - \EEx{y_1,y_2}{A_i(F_1(y_1),F_2(y_2))}\right| \le\frac{1}{\poly(n)},\]
}
as desired. 
\end{proof} 

\def\shortbib{0}
\bibliographystyle{alpha}
\bibliography{references}
\end{document}